\documentclass[reprint,amsmath,amssymb,aps,nofootinbib,prd,superscriptaddress]{revtex4-2}
\usepackage{custom}
\def\e{\mathrm{e}}
\def\D{\mathrm{d}}

\def\eq#1{\begin{equation}\begin{aligned}#1\end{aligned}\end{equation}}
\def\eref#1{(\ref{eq:#1})}

\begin{document}

\title{Interacting quantum reference frames}

\author{Leonardo Chataignier}
\affiliation{Centro Brasileiro de Pesquisas Físicas, Rua Dr. Xavier Sigaud 150, Urca, CEP: 22290-180, Rio de Janeiro, RJ, Brazil}
\email{lchataig@cbpf.br}
\author{Alexander R. H. Smith}
\affiliation{Department of Physics, Saint Anselm College, Manchester, New Hampshire 03102, USA}
\email{arhsmith@anselm.edu}
\affiliation{Department of Physics and Astronomy, Dartmouth College, Hanover, New Hampshire 03755, USA}

\date{\today}

\begin{abstract}
We formulate quantum reference frames within the Dirac quantization of gauge theories, where physical states are obtained from a kinematical Hilbert space $\mathcal{H}_{\rm kin}$ by group averaging. We assume the gauge group $G$ is unimodular, and we define a quantum reference frame as a system of coherent states associated with a group $G_R$ of frame orientations acting on a tensor factor of $\mathcal{H}_{\rm kin}$. In direct analogy with the Faddeev--Popov procedure in path integrals, we show that the physical inner product can be written in terms of a choice of frame orientation with an appropriate operator measure, but that the physical Hilbert space is independent of this choice of frame. This yields a generalization of the Page--Wootters reduction map used to define conditional states and conditional probabilities of measurement outcomes given a frame orientation. These probabilities correspond to the expectation values of relational Dirac observables obtained from generalizations of previously considered relativization maps. Our main result is that the map that implements a change of frame is unitary, even when $G$ does not act as a tensor-product representation across a factorization of $\mathcal{H}_{\rm kin}$. The formalism therefore accommodates interactions between reference frames and the systems described relative to them. We apply this formalism to temporal frames, showing that a unitary Schr\"odinger equation emerges for an interacting clock and system, and that a previously reported time-nonlocal Schr\"odinger equation is recovered in an appropriate limit. With this, problems of non-unitarity are avoided by the use of the gauge-fixing measure in the physical inner product. We also discuss the case of interacting spatial frames, where the known correspondence between coherence in one frame and entanglement in another is shown to depend on the coupling.
\end{abstract}

\maketitle

\section{\label{sec:intro}Introduction}
A reference frame is a physical system relative to which observations and measurements are made. Since any physical system is ultimately quantum mechanical, reference frames should themselves be described by quantum theory: they may be delocalized, prepared in superpositions, or entangled with the systems they describe. Indeed, the motivation for considering quantum (as opposed to classical) reference frames comes from several directions. Experiments now routinely probe the quantum properties of microscopic systems and correlations between them, opening the possibility of describing quantum systems relative to one another rather than to macroscopic, classical rods and clocks \cite{zych_quantum_2011,moreva_time_2014,pikovski_universal_2015,roura_gravitational_2020,grochowski_quantum_2021,paczos_quantum_2024,hu__feasibility_2024,cepollaro_quantum_2024,sorci_quantum_2026,OConnell_2010,Hornberger_2012,Aspelmeyer_2014,Frowis_2018,Fein_2019,Millen_2020,Bild_2023,Bose_2025,Marletto_2025}. In addition, general relativity leads to the expectation that a theory of quantum gravity should not depend on extrinsic standards of time and space but rather on relational notions built from the degrees of freedom of the theory, thus highlighting the need to understand quantum reference frames~\cite{dewitt_quantum_1967,rovelli_quantum_1991, rovelli_quantum_2004,hardy_construction_2018,giacomini_quantum_2022,freidel_gravitational_2026,chen_quantum_2026,aguilar-gutierrez_relational_2026} and, in particular, quantum clocks~\cite{giovannetti_quantum_2015, ruiz_entanglement_2017,loveridge_relative_2019,marletto_evolution_2017, castro-ruiz_quantum_2020,smith_quantum_2020,Chataignier:2019,hohntrinity2021,Chataignier:2020,hohn_equivalence_2021,giacomini_spacetime_2021,chataignier_relational_2026,gielen_singularity_2020,foti_time_2021,cafasso_quantum_2024,gielen_unitarity_2022,vishal_singularity_2026}. Quantum reference frames can also be used in more general gauge theories, where relevant observables can be constructed relationally, as gauge-invariant quantities ~\cite{rovelli_why_2014,hohn_matter_2024,fedida_groupoidal_2026,kondo_understanding_2026,lacambra_gauss_2026}. Likewise, quantum reference frames play a role in quantum information theory, where information may be encoded, processed, and communicated in relational degrees of freedom~\cite{bartlett_classical_2003,bartlett_reference_2007,gour_resource_2008,bartlett_quantum_2009,ahmadi_communication_2015,smith_communicating_2019,carrozza_correspondence_2024}, and are essential to the foundations of quantum theory itself~\cite{aharonov_charge_1967,poulin_toy_2006,angelo_physics_2011,smith_quantum_2016,loveridge_symmetry_2018, giacominiQuantumMechanicsCovariance2019,hamette_quantum_2020,delahametteperspectiveneutral2021,ali_ahmad_quantum_2022,castro-ruiz_relative_2025,jorquera_riera_uncertainty_2025}.

The inclusion of interactions between reference frames and other systems is paramount, since they are unavoidable when gravity is present. Much of the recent literature on quantum reference frames is based on the celebrated formalism of Page and Wootters~\cite{page_evolution_1983,wootters_time_1984}, which does not, however, contemplate interactions.\footnote{Page does consider interacting spin systems in Ref.~\cite{page_clock_1994}, albeit the global state is assigned through an incoherent $G$-twirl and not by projecting onto a physical Hilbert space as explored here and developed in later works on the Page-Wootters formalism.} While there have been many proposed generalizations of the Page--Wootters formalism to the interacting case, there have been difficulties in obtaining a unitary relational theory~\cite{smith_quantizing_2019, paiva_non-inertial_2022, paiva_flow_2022, cohen_quantum_2023, mendes_non-linear_2021, rijavec_robustness_2023, gemsheim_emergence_2023,ludescher_quantum_2026,rijavec_conditions_2026,shimada_nonunitarity_2026}, and thereby in interpreting the resulting theories.

It is interesting to note that similar problems of non-unitarity occur in gauge theories if one neglects the proper gauge-fixing measure. Feynman famously introduced what he called ``artificial, dopey particles'' to restore unitarity in diagrams of perturbative gravitational processes~\cite{Feynman63}. Later, such particles became known as the Feynmann--DeWitt--Faddeev--Popov ghosts, which were understood to define the proper gauge-fixing measure in the path integral~\cite{Teitelboim1991,HT:book,Barvinsky_1993}.

In this article, we formulate quantum reference frames in the context of gauge theories and show that doing so leads to a unitary relational theory even when the reference frame interacts with a system that is described relative to it. By analogy with general relativity, we take the view that there are no preferred quantum reference frames, so that the physical content of a state cannot depend on which system is used as a reference frame. We encode this frame-independence by imposing an invariance condition on states of a composite system: physical sates are those that are invariant under the action of a gauge group~$G$. Choosing a reference frame then amounts to gauge fixing, so that each choice, together with its associated gauge-fixed inner product, yields a description of one subsystem relative to another. Because gauge fixing is a matter of description rather than of physics, the description obtained in different frames must carry the same invariant content. This requirement is expressed by the unitarity of the change-of-frame map established below.

We begin in Sec.~\ref{Sec:InteractingQRF} by reviewing the quantization of a gauge theory associated with a locally compact, unimodular group~$G$ and the construction of the gauge-invariant physical Hilbert space $\mathcal{H}_{\rm phys}$ from a kinematical Hilbert space $\mathcal{H}_{\rm kin}$ via group averaging. We then introduce quantum reference frames, defined as a coherent-state system associated with the group of reference frames orientations~$G_R$ that acts on a factor $\mathcal{H}_{R}$ of the kinematical Hilbert space, $\mathcal{H}_{\rm kin} \simeq\mathcal{H}_{R} \otimes \mathcal{H}_{S}$, where $\mathcal{H}_{S}$ describes a subsystem $S$ comprised of everything other than the frame $R$. Employing such reference frames allows us to gauge fix the physical inner product in analogy with  the Faddeev--Popov procedure  in path integrals. This induces conditional states of $S$ relative to $R$ through a generalization of the Page--Wootters reduction map.

Our main technical result follows: we construct a change-of-frame map from this generalized reduction map and prove that it is unitary. Crucially, unitarity does not require the gauge group $G$ to act as a tensor product representation across the factorization of the kinematical Hilbert space, $\mathcal{H}_{\rm kin} \simeq \mathcal{H}_{R} \otimes \mathcal{H}_{S}$, into reference frame and subsystem of interest (see also Ref.~\cite{ludescher_quantum_2026}), and may not coincide with $G_R$. This generalization accounts for the possibility of interactions between $S$ and $R$, and reduces to the known change-of-frame maps in the non-interacting case~\cite{giacominiQuantumMechanicsCovariance2019,hohntrinity2021,delahametteperspectiveneutral2021,hamette_quantum_2020}.

In Sec.~\ref{Sec:Relational Dirac observables} we use the generalized reduction map and its inverse to construct relational Dirac observables on $\mathcal{H}_{\rm phys}$. The construction is similar to earlier proposals by one of us \cite{Chataignier:2019,Chataignier:2020}, and  generalizes previously discussed relativization maps \cite{bartlett_reference_2007,loveridge_symmetry_2018,hohntrinity2021}. We show that the expectation value of these observables are conditional probabilities associated with a measurement of $S$ having a particular outcome conditioned on $R$ taking a particular orientation $r\in G_R$. Sections~\ref{Sec:TemporalFrames} and \ref{Sec:ChangeOfSpatialFrame} apply the developed formalism. In Sec.~\ref{Sec:TemporalFrames} we treat temporal frames and show that a unitary Schr\"{o}dinger equation emerges even when the clock interacts with a system that is described relative to it; the time-nonlocal Schr\"{o}dinger equation of Ref.~\cite{smith_quantizing_2019} is recovered in an appropriate limit. In Sec.~\ref{Sec:ChangeOfSpatialFrame} we apply the change-of-frame map to spatial reference frames that are coupled to one another, and show that the known correspondence between coherence in one frame and entanglement in another frame~\cite{giacominiQuantumMechanicsCovariance2019,ali_ahmad_quantum_2022,de_la_hamette_entanglement-asymmetry_2022} depends on this coupling. We conclude in Sec.~\ref{Sec:ConclusionAndOutlook} with a summary of results and outlook on future directions.

\section{Interacting quantum reference frames}
\label{Sec:InteractingQRF}
\subsection{The kinematical and physical Hilbert spaces}
The statistics of a quantum observable are associated with a sample space $\Omega$ defined as the set of all possible measurement outcomes, an event space $\mathcal{F}$ defined as a $\sigma$-algebra of subsets of~$\Omega$, and a probability measure on $\mathcal{F}$. The probability of an event $X \in \mathcal{F}$ is given by the Born rule
\begin{equation}\label{eq:Born}
    P(X) = \tr \left[ \rho E(X)\right],
\end{equation}
where $\rho \in \mathcal{S}(\mathcal{H})$ is a density operator (i.e., positive-semidefinite operator with unit trace) on a Hilbert space $\mathcal{H}$ describing the state of a system of interest, and $E(X) \in \mathcal{E}(\mathcal{H})$ is the effect operator (i.e., a bounded self adjoint operator such that $0\leq E(X) \leq I$) associated with the event $X$. The effect operators constitute a positive operator-valued measure (POVM) associated with the observable, defined as the map~\cite{hall_quantum_nodate,busch_quantum_2016}
\begin{align}
    E \, : \, \mathcal{F} \to \mathcal{E}(\mathcal{H}), \
    X \mapsto E(X),
\end{align}
where the effect operators $E(X)$ must  1) be positive semi-definite  for all $X \in \mathcal{F}$, 2) sum to the identity, $E(\Omega) = I$, and 3) satisfy $E(\cup_i X_i) = \sum_i E(X_i)$ for disjoint subsets $X_i\in\mathcal{F}$.\footnote{If the effect operators are projective, $E(X)^2 = E(X)$ for all $X\in \mathcal{F}$, the POVM is said to be a `projection-valued measure' (PVM) and the associated observable is said to be `sharp.' Through the spectral theorem, a PVM defines a unique self-adjoint operator, and conversely, a self-adjoint operator defines a PVM. However, there is no {\it a priori} reason to restrict to observables associated with self-adjoint operators and PVMs, and indeed, general reference frames are described by more general POVMs.} These properties ensure that $P(X)$ in Eq.~\eqref{eq:Born} is a well-defined probability distribution.

Now suppose that events related by a group $G$ of gauge transformations represent the same physical situation. The probability of an event $X$ should thus be equivalent to the probability of the transformed event $gX$ for all $g\in G$. We are led to an equivalence class $X\sim gX$, and the choice of representative can be thought of as a choice of gauge. As the theory does not single out a preferred choice, this is a redundancy that might refer to the choice of dynamical reference frames in spacetime (e.g., clocks and rods) or to more abstract quantities such as a choice of gauge in electromagnetism. We will refer to $G$ as the `gauge group' and to $g\in G$ as a `gauge transformation.' For simplicity, we assume that $G$ is locally compact and unimodular, so that the left and right Haar measures coincide, being denoted by $\D g$ with suitable normalization constants absorbed.

Let the Hilbert space $\mathcal{H}$ carry a unitary representation of the gauge group, $U: G \to \mathcal{B}(\mathcal{H}), \, g \mapsto U(g)$. Given the equivalence relation $X \sim gX$, it follows that the effect operators associated with the events $X$ and $gX$ are related by $E(gX) = U(g)E(X)U(g)^\dagger$; such POVMs are said to be covariant with respect to $G$~\cite{holevo_probabilistic_1982, busch_quantum_2016}. For the probability distribution in Eq.~\eqref{eq:Born} to be invariant under such a gauge transformation, due to cyclic property of the trace, the state of the system must be invariant,
\begin{equation}
    \rho = U(g)^\dagger \rho U(g) \, \iff \,  [\rho, U(g)] = 0,
    \label{eq:invarainceOfState}
\end{equation}
for all $g\in G$.

The spectrum of a unitary operator can be decomposed into its point spectrum and continuous spectrum.  If the continuous spectrum of $U(g)$ is empty, then there exist invariant states $\rho \in \mathcal{S}(\mathcal{H})$ satisfying Eq.~\eqref{eq:invarainceOfState}. In this case, $\rho$ and $U(g)$ can be diagonalized in the same eigenbasis, and if $\rho$ is pure, it can be written as $\rho = \ket{\Psi} \!\bra{\Psi}$, where $\ket{\Psi}$ is an eigenstate of $U(g)$.

However, if the point spectrum of $U(g)$ is empty, then no invariant states exist. In this case, or when neither the point nor the continuous spectrum is empty, one can rig the Hilbert space $\mathcal{H}$ by introducing a dense subspace $\Phi \subset \mathcal{H}$ and its dual space $\Phi^*$, such that $\Phi \subset \mathcal{H} \subset  \Phi^*$ is a Gelfand triple. By the generalized spectral theorem~\cite{hall_quantum_nodate, busch_quantum_2016}, there exist elements of $\Phi^*$ that are invariant under the action of $U(g)$ up to an overall phase. We will neglect this phase and consider only the elements of $\Phi^*$ on which the gauge group acts trivially, that is,
\eq{\label{eq:U-invariance}
U(g)\ket{\Psi}=\ket{\Psi} \text{ for } \ket{\Psi} \in \Phi^*,} 
which defines a set of gauge-invariant generalized states.\footnote{This is equivalent to the Dirac quantization of a constrained system in which classical constraints are promoted to self-adjoint operators $C_\mu$ acting on $\mathcal{H}_{\rm kin}$ and demanding physical states lie in the kernel of these constraint operators. This can be seen by taking a derivative of Eq.~\eqref{eq:U-invariance}, which yields $C_\mu\ket{\Psi}\ce-i\partial/\partial g^\mu|_{g = 0} U(g)\ket{\Psi} = 0$, where $g^\mu$ are coordinates in the group manifold.} This subset of invariant generalized states can be turned into a Hilbert space by introducing an inner product between elements of this set and completing it with respect to the norm induced by this inner product (e.g., Refs.~\cite{marolf_group_2002, rovelli_quantum_2004,thiemann_modern_2008}). The resulting Hilbert space is called the physical Hilbert space $\mathcal{H}_{\rm phys}$ and the Hilbert space from which it was constructed is called the kinematical Hilbert space $\mathcal{H}_{\rm kin} \simeq \mathcal{H}$.

The physical inner product can be defined in terms of the kinematical inner product via group averaging~\cite{marolf_group_2002,thiemann_modern_2008}. Specifically, the inner product between $\ket{\Psi_1}$ and $\ket{\Psi_2}$ in $\mathcal{H}_{\rm phys}$ can be expressed in terms of the kinematical inner product between $\ket{\psi_1}$ and $\ket{\psi_2}$ in $\mathcal{H}_{\rm kin}$ as 
\begin{equation}
    \braket{\Psi_1 | \Psi_2}_{\rm phys} \ce \bra{\psi_1} \int_G \D g \, U(g) \ket{\psi_2}_{\rm kin}, 
\label{eq:PhysInnerProduct}
\end{equation}
where
\begin{equation}\label{eq:state-averaging}
    \ket{\Psi_i} = \int_G \D g \, U(g) \ket{\psi_i} \text{ for } i\in \{1,2\} .
\end{equation}
If the continuous spectrum of $U(g)$ is empty, this inner product coincides with the kinematical inner product of two physical states of the form given in Eq.~\eqref{eq:state-averaging}, so it can be adopted as a general definition of the inner product for a theory with a locally compact, unimodular gauge group $G$.

\subsection{Quantum reference frames}
A reference frame is a subsystem relative to which other subsystems are described. With this in mind, let us partition the kinematical Hilbert space into two subsystems $\mathcal{H}_{\rm kin} = \mathcal{H}_R\otimes\mathcal{H}_S$, where $\mathcal{H}_R$ is the Hilbert space of a chosen reference frame $R$, and $\mathcal{H}_S$ is the Hilbert space of a subsystem  $S$ comprised of everything else. 

Associated with any reference frame is a group of frame transformations $G_R$ whose elements $r \in G_R$ correspond to different orientations of the reference frame; for example, in the case of a spatial reference frame, $G_R$ would be the translation group and $r$ would label different positions (i.e., orientations) of the reference frame. We assume for simplicity that $G_R$ is a Lie group, although the formalism to be developed could also be applied to discrete groups. Suppose that $\mathcal{H}_R$ carries a unitary representation of the group of frame transformations, $U_R: G_R \to \mathcal{H}_R, \ r \mapsto U_R(r)$. Associated with this representation is a set of group coherent states $\{\ket{r} \ce U_R(r) \ket{r_0} \ | \ r\in G_R\}$, where $\ket{r_0}$ is chosen so that this set forms a (possibly overcomplete) basis for $\mathcal{H}_R$~\cite{perelomov_generalized_1986,zhang_coherent_1990}; depending on whether $G_R$ is compact or noncompact, the states $\ket{r}$ may or may not be elements of~$\mathcal{H}_R$. Given that these coherent states are labeled by $r\in G_R$, each coherent state is associated with a different orientation of the reference frame. 

A measurement of the reference frame orientation is then described by a POVM with outcome space $\Omega_R = G_R$ and effect densities $\{E_R(r) \ce \ket{r}\!\bra{r}  \ | \  r \in G_R\}$. Due to the definition of the group coherent states, this reference frame orientation POVM is covariant with respect to $G_R$:
\eq{
E_R(rs)=U_R(r)E_R(s)U_R(r)^\dagger, \ \forall \ r,s \in G_R.
} 
Such group covariant POVMs are known to be optimal estimators of group transformations~\cite{holevo_probabilistic_1982,braunstein_statistical_1994,braunstein_generalized_1996,wiseman_quantum_2010,haapasalo_optimal_2021}, such as the orientation $r$ of the reference frame $R$.

In general, the representation of the gauge group $U(g)$ will not be a tensor product representation across the partition $\mathcal{H}_{\rm kin} = \mathcal{H}_R\otimes\mathcal{H}_S$, that is, $U(g) \neq V_R(g) \otimes V_S(g)$, where $V_R(g)$ and $V_S(g)$ are unitary representations of $g\in G$ on $\mathcal{H}_R$ and $\mathcal{H}_S$, respectively; we will refer to this as the interacting case because gauge transformations can entangle the reference frame and system across the kinematical partition. If it happens to be the case that $U(g) = V_R(g) \otimes V_S(g)$, then indeed $U(g)$ is a tensor product representation of $G$; we will refer to this as the non-interacting case. In this case, $V_R(g)$ may not be a representation of the group of reference frames $g\in G_R$, that is, $G\neq G_R$; see Refs.~\cite{vanrietvelde_switching_2023,ali_ahmad_relational_2025} for examples. It may be desirable to choose the reference frame and associated factorization of the kinematical Hilbert space so that the representation of the gauge group is a tensor product representation; however, this may not always be possible or describe the physical situation of interest.

\subsection{Gauge fixing and conditional states}
We are interested in a gauge-invariant description of the subsystem $S$ relative to a reference frame $R$ with an orientation $r\in G_R$. To do so, we will express the physical inner product in terms of the kinematical inner product and a reference frame orientation $r\in G_R$, which will be seen to induce conditional states of $S$ conditioned on $R$ having orientation $r$ via a generalization of the Page--Wootters reduction map.

To accomplish this, we will first derive an identity analogous to the Faddeev--Popov resolution of the identity that is used to regularize gauge-theoretic path integrals. We can extract the gauge-invariant content of the kinematical operator $E_R(r)\otimes I_S$, where $I_S$ is the identity on $\mathcal{H}_S$, by the so-called $G$-twirl over the gauge group $G$ 
\eq{\label{eq:Gtwirl}
\mathcal{G}[E_R(r)\otimes I_S]\ce\int_G\D g\, U(g)\left(E_R(r)\otimes I_S\right)U(g)^{\dagger} .
}
One can verify that $\mathcal{G}[E_R(r)\otimes I_S]$ commutes with $U(g)$, and is thus gauge invariant. In what follows, we consider the case when $\mathcal{G}[E_R(r)\otimes I_S]$ is a well-defined and invertible operator on $\mathcal{H}_{\rm kin}$.\footnote{The invertibility of this operator will depend on the spectral properties of the $U(g)$ and $E_R(r)$ such that its kernel is trivial. See the Sec.~\ref{Sec:ConclusionAndOutlook} for a possible connection to the Gribov problem.} Its inverse $\Delta_R(r)$ is strictly positive, and thus it has a unique square root that commutes with $U(g)$. Then,
\eq{\label{eq:gf}
I_{\rm kin} &= \Delta_R(r) \mathcal{G}[E_R(r)\otimes I_S] \\
&=\Delta^{\frac12}_R(r)\mathcal{G}[E_R(r)\otimes I_S]\Delta^{\frac12}_R(r)\\
&= \mathcal{G}[\Delta^{\frac12}_R(r)(E_R(r)\otimes I_S)\Delta^{\frac12}_R(r)]\\
&=:\mathcal{G}[\Xi_R(r)],
}
where $I_{\rm kin}$ is the identity on $\mathcal{H}_{\rm kin}$ and 
\eq{
\Xi_R(r)\ce\Delta^{\frac12}_R(r)(E_R(r)\otimes I_S)\Delta^{\frac12}_R(r). }
This is the analogue of the usual Faddeev--Popov resolution of the identity
\eq{\label{eq:FP}
1 = \boldsymbol{\Delta}_{\rm FP}[A]\int\mathcal{D}\varepsilon\,\delta\!\left(\chi[A^\varepsilon]\right),
}
where $A$ is a gauge field and $A^{\varepsilon}$ is its gauge-transformed counterpart, $\chi[A^{\varepsilon}]$ is a gauge-fixing condition, $\int\mathcal{D}\epsilon$ is a functional integration over the gauge group, and $\boldsymbol{\Delta}_{\rm FP}[A]$ is the absolute value of the so-called Faddeev--Popov determinant, which can be expanded in terms of ghost variables (e.g., \cite{Teitelboim1991,HT:book,Barvinsky_1993}). In Eq.~\eref{gf}, $E_R(r)\otimes I_S$ is the analogue of the gauge-fixing delta distribution $\delta\!\left(\chi[A^\varepsilon]\right)$, the $G$-twirl is the analogue of the functional gauge-group integration, and $\Delta_R(r)$ is the analogue of the absolute value of the Faddeev--Popov determinant. 

By construction, the operator $\Xi_R(r)$ serves as a gauge-fixing measure on the physical inner product, since we can use Eqs.~\eqref{eq:PhysInnerProduct},~\eqref{eq:Gtwirl} and~\eqref{eq:gf} to express the physical inner product in terms of the kinematical inner product:
\eq{\label{eq:PhysInnerProduct2}
    \braket{\Psi_1| \Psi_2}_{\rm phys} &\ce \bra{\psi_1 } \int_G \D g \, U(g) \ket{\psi_2}_{\rm kin} \\
    &= \braket{\psi_1|\Psi_2}_{\rm kin}\\
      &=  \bra{\psi_1} \int_G dg \,U(g) \Xi_R(r)  U(g)^\dagger \ket{\Psi_2}_{\rm kin} \\
      &=  \left( \bra{\psi_1} \int_G \D g \, U(g)^\dagger\right) \Xi_R(r)  \ket{\Psi_2}_{\rm kin} \\
      &= \bra{\Psi_1 } \Xi_R(r) \ket{ \Psi_2}_{\rm kin}.
}
Thus, one can insert $\Xi_R(r)$ in the kinematical inner product to obtain the physical overlap of two elements of~$\mathcal{H}_{\rm phys}$. Notice, however, that this overlap is by construction independent of the choice of reference frame and, in particular, does not depend on the frame's orientation~$r\in G_R$.

In general, the effect densities $E_R(r)\otimes I_S$ for all $r \in G_R$ do not form a covariant POVM with respect to the gauge group $G$. However, for every $r\in G_R$, $\Xi_R(r)$ is the seed of a POVM that is covariant with respect to $G$, specifically, a POVM with outcome space $\Omega = G$ and effect operators
\begin{equation}
    E(X) = \int_{X\subseteq G} \D g \, U(g)\Xi_R(r)U(g)^{\dagger}, 
\end{equation}
that satisfy the covariance condition
\begin{equation}
E(gX) = U(g) E(X) U(g)^\dagger.  
\end{equation}
In this way, $\Xi_R(r)$ can be thought of as a `gauge covariantization' of the effect density $E_R(r).$ 

In the non-interacting case when $G =G_R$ and $U(g) = U_R(g) \otimes U_S(g)$, the effect densities $E_R(r)\otimes I_S$ form a POVM that is covariant with respect to $G$. In this case, Eq.~\eref{Gtwirl} becomes
\eq{
\mathcal{G}[E_R(r)\otimes I_S]\ce\int_G\D g\, E_R(gr)\otimes I_S = I_{\rm kin} ,
}
and $\Delta_R(r) = I_{\rm kin}$ is simply the identity on $\mathcal{H}_{\rm kin}$. Then, Eq.~\eqref{eq:PhysInnerProduct2} takes the form
\eq{\label{eq:PhysInnerProduct3}
    \braket{\Psi_1| \Psi_2}_{\rm phys} &
      &= \bra{\Psi_1 } \left(\ket{r}\!\bra{r} \otimes I_S \right)\ket{ \Psi_2}_{\rm kin},
}
which coincides with the conditional inner product introduced in Refs.~\cite{smith_quantizing_2019,hohntrinity2021,delahametteperspectiveneutral2021}.

The form of the physical inner product in Eq.~\eqref{eq:PhysInnerProduct2} suggests defining conditional states of the system $S$ relative to the reference frame $R$ having the orientation $r$ as
\begin{equation}
    \ket{\psi_{S}(r)} \ce  \bra{r} \otimes I_S \Delta_R^{\frac{1}{2}}(r) \ket{\Psi} = \mathcal{R}(r) \ket{\Psi},
    \label{eq:conditionalState}
\end{equation}
where we have introduced the reduction map
\eq{
        \mathcal{R}(r): \mathcal{H}_{\rm phys} \to \mathcal{H}_{S|R}, \ 
        \ket{\Psi} \mapsto \ket{\psi_S(r)},
        \label{eq:reducationMap}
}
where $\mathcal{H}_{S|R}$ is the Hilbert space of $S$ relative to $R$, defined as the image of $\mathcal{R}(r)$; as shown in Appendix~\ref{SubspaceAppendix}, $\mathcal{H}_{S|R} \subseteq \mathcal{H}_S$.

It follows that the overlap between two physical states can be expressed as the overlap between two conditional states of $S$ relative $R$ as 
$\braket{\Psi_1|\Psi_2}_{\rm phys} = \braket{\psi_{S,1}(r) | \psi_{S,2}(r)}$, where $\ket{\psi_{S,i}(r)} = \mathcal{R}(r) \ket{\Psi_i}$ for $i\in\{1,2\}$. The conditional states in Eq.~\eqref{eq:conditionalState} and the formalism that follows constitute a generalization of the Page--Wootters formalism~\cite{page_evolution_1983,wootters_time_1984} to unimodular gauge groups $G$ and a possibly distinct group of reference frames $G_R$ along the lines of Refs.~\cite{smith_quantizing_2019,hohntrinity2021,delahametteperspectiveneutral2021}.

The reduction map is invertible, with left-inverse 
\begin{equation}
    \mathcal{R}^{-1}(r) \ce \int_G\D g\, U(g) \Delta_R^{\frac12}(r)\ket{r}\otimes I_S,
\end{equation}
which can be verified by direct computation:
\begin{equation}
\begin{aligned}
    &\mathcal{R}^{-1}(r) \circ \mathcal{R}(r)\ket{\Psi}\\
    &= \int_G\D g\,U(g)\Delta_R^{\frac12}(r)\ket{r}\!\bra{r}\otimes I_S \Delta_R^{\frac12}(r)\ket{\Psi} \\
    &= \int_G\D g\,U(g)\Delta_R^{\frac12}(r)\ket{r}\!\bra{r}\otimes I_S \Delta_R^{\frac12}(r)U(g)^\dagger\ket{\Psi} \\
    &=\ket{\Psi},
    \label{eq:Iphys}
\end{aligned}
\end{equation}
where the second equality follows from the fact that the domain of $\mathcal{R}(t)$, which is the physical Hilbert space $\mathcal{H}_{\rm phys}$, is invariant under the action of $U(g)^\dagger$, and the last equality follows from the definition of $\Delta_R(r)$. Given that the reduction map is invertible, it follows that the physical Hilbert space and Hilbert space of $S$ relative to $R$ are isomorphic, $\mathcal{H}_{\rm phys} \simeq \mathcal{H}_{S|R}$. Moreover, using the identity above and the definition of the reduction map, a physical state is seen to be entangled across the $\mathcal{H}_R\otimes\mathcal{H}_S$ partition of the kinematical Hilbert space:
\begin{equation}
    \ket{\Psi}=\int_G \D g \,U(g) \Delta_R^{\frac{1}{2}}(r) \ket{r}\otimes\ket{\psi_{S}(r)}.
\end{equation}

We now show that the norm of a conditional state of $S$ relative to $R$ is orientation-independent and determined by the norm of the physical state. This is important because, to preserve the probabilistic interpretation of the conditional state, its norm should not depend on the orientation of $R$.
\begin{theorem}
\label{theorem:isometry}
    The norm of the conditional state $\ket{\psi_S(r)} = \mathcal{R}(r)\ket{\Psi}$ is $r$-independent and equal to of the norm of the associated element of $\mathcal{H}_{\rm phys}$, $|\psi_S(r)| = |\Psi|$.
\end{theorem}
\begin{proof}
     Consider a conditional state of the system  defined as $\ket{\psi_S(r)} = \mathcal{R}(r) \ket{\Psi}$, with $\ket{\Psi}\in\mathcal{H}_{\rm phys}$. Observe that
\begin{equation}
    \begin{aligned} \abs{\psi_S(r)}^2 &= \braket{\psi_S(r)|\psi_S(r)} \\
    &=\bra{\Psi} \mathcal{R}(r)^\dagger \mathcal{R}(r) \ket{\Psi}_{\rm kin} \\ 
    &=\bra{\Psi} \Delta_R^{\frac{1}{2}}(r) \ket{r}\!\bra{r} \otimes I_S \Delta_R^{\frac{1}{2}}(r) \ket{\Psi}_{\rm kin} \\  &=\braket{\Psi|\Psi}_{\rm phys} \\
    &=\abs{\Psi}^2,
    \end{aligned}
\end{equation}
from which it follows that the norm of $\ket{\psi_S(r)}$ is $r$-independent and equal to the norm of $\ket{\Psi}$. 
\end{proof}

\subsection{Changing quantum reference frames}
Now consider a different partition of of the kinematical Hilbert space, $\mathcal{H}_{\rm kin} = \mathcal{H}_{\tilde{R}}\otimes\mathcal{H}_{\tilde{S}}$, into a new system $\tilde{S}$ and reference frame $\tilde{R}$ associated with the reduction map $\tilde{\mathcal{R}}(\tilde{r})$. We can compose the reduction maps of different reference frames to define a change-of-frame map as
\begin{equation}
    \ket{\psi_{\tilde{S}}(\tilde{r})} = \tilde{\mathcal{R}}(\tilde{r}) \circ \mathcal{R}^{-1}(r) \ket{\psi_S(r)}.
\end{equation}
In the particular case in which $\tilde{R}$ and $R$ are the same frame, then this map corresponds to a change of frame orientation, $r\to\tilde{r}$.

We now show that the change-of-frame map is unitary, even in the presence of interactions (i.e., $\Delta_R(r) \neq I_{\rm kin}$).
\begin{figure}
    \centering   \includegraphics[width=\linewidth]{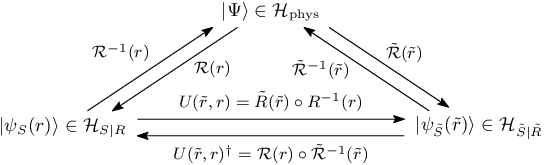}
    \caption{The reduction maps $\mathcal{R}(r)$ and $\tilde{\mathcal{R}}(\tilde{r})$ map physical states $\ket{\Psi} \in \mathcal{H}_{\rm phys}$ to conditional states  $\ket{\psi_S(r)} \in \mathcal{H}_{S|R}$ and $\ket{\psi_{\tilde{S}}(\tilde{r})} \in \mathcal{H}_{\tilde{S}|\tilde{R}}$ of $S$ and $\tilde{S}$ relative to the reference frames $R$ and $\tilde{R}$, respectively. By composing these maps and their inverses, the change-of-frame map $U(\tilde{r},r)$ is constructed, which maps conditional states relative to one reference $R$ and orientation $r$ to conditional states relative to a different reference frame $\tilde{R}$ and orientation $\tilde{r}$.
    } 
    \label{fig:ChangeOfFrame}
\end{figure}

\begin{theorem}
\label{thm:two}
    The change-of-frame map 
    \begin{equation}
        \begin{aligned}
            U(\tilde{r},r)\ce \tilde{\mathcal{R}}(\tilde{r}) \circ \mathcal{R}^{-1}(r)
        \label{eq:change-frame-map}
        \end{aligned}
    \end{equation} is unitary with respect to the inner products on $\mathcal{H}_{S|R}$ and $\mathcal{H}_{\tilde{S}|\tilde{R}}$, which are induced by the physical inner product.
\end{theorem}
\begin{proof} 
Observe that
\eq{\notag
U(\tilde{r},r) = \int_G \D g \, \bra{\tilde{r}} \otimes I_{\tilde{S}}\Delta_{\tilde{R}}^{\frac{1}{2}}(\tilde{r})U(g) \Delta_R^{\frac{1}{2}}(r) \ket{r} \otimes I_S ,
}
and thus
\begin{align}
        U(\tilde{r},r)^\dagger &= \int_G \D g \, \bra{r} \otimes I_S\Delta_{R}^{\frac{1}{2}}(r)U(g)^\dagger \Delta_{\tilde{R}}^{\frac{1}{2}}(\tilde{r}) \ket{\tilde{r}} \otimes I_{\tilde{S}} \notag\\
        & = \mathcal{R}(r)\circ \tilde{\mathcal{R}}^{-1}(\tilde{r}) \label{eq:Udagger}\\
        &= U(r,\tilde{r}).\notag
\end{align}
It then follows that
\begin{equation}
\begin{aligned}\notag
      U(\tilde{r},r)^\dagger U(\tilde{r},r)
      &= \mathcal{R}(r) \!\circ\! \tilde{\mathcal{R}}^{-1}(\tilde{r}) \!\circ \!\tilde{\mathcal{R}}(\tilde{r}) \!\circ\! \mathcal{R}^{-1}(r)\\
      &= \mathcal{R}(r) \circ  \mathcal{R}^{-1}(r)\\
      &= I_{S|R},
\end{aligned}
\end{equation}
where the second equality follows from Eq.~\eqref{eq:Iphys} and the fourth equality from the observation that
\eq{\label{eq:rightInverse}
        \mathcal{R}(r) \circ \mathcal{R}^{-1}(r) \ket{\psi_S(r)} &= \mathcal{R}(r) \circ \mathcal{R}^{-1}(r) \circ \mathcal{R}(r) \ket{\Psi} \\
   &= \mathcal{R}(r) \ket{\Psi} \\
&= \ket{\psi_S(r)}.
}
This establishes that $U(\tilde{r},r)$ is an isometry, consistent with Theorem~\ref{theorem:isometry}. Similarly, we find
\begin{equation}
    \begin{aligned}\notag
        U(\tilde{r},r) U(\tilde{r},r)^\dagger = I_{\tilde{S}|\tilde{R}}.
    \end{aligned}
\end{equation}
Therefore, $U(\tilde{r},r)^\dagger   =U^{-1}(\tilde{r},r)$ is the left and right inverse of the map $U(\tilde{r},r)$, which is seen to be a unitary operator.
\end{proof}

\begin{corollary*}
    When $R$ and $\tilde{R}$ coincide the change-of-frame map changes the orientation of the state of $S$ relative $R$ from $r$ to~$r'$:
    \begin{equation}
        \ket{\psi_S(r')} = U(r',r) \ket{\psi_S(r)}.
        \label{eq:OrientationChange}
    \end{equation}
    Moreover, $U(r,r) = I_{S|C}$.
\end{corollary*}
\begin{proof}
    Equation~\eqref{eq:OrientationChange} follows from the definition of the reduction map and its inverse:
    \eq{
    U(r',r)\ket{\psi_S(r)} &= \left[ \mathcal{R}(r') \circ \mathcal{R}^{-1}(r) \right] \left[ \mathcal{R}(r) \ket{\Psi}\right] \\
    &= \mathcal{R}(r')  \ket{\Psi} \\
    &=\ket{\psi_S(r')}.\notag
    }
    Moreover, $U(r,r)=I_{S|C}$ follows immediately from Eq.~\eqref{eq:rightInverse}.
\end{proof}
The change-of-frame map is a generalization of the change-of-frame map defined in Ref.~\cite{hohntrinity2021} for one-dimensional gauge groups (see Refs.~\cite{hohn_equivalence_2021,hamette_quantum_2020,delahametteperspectiveneutral2021, hohn_matter_2024} for further generalizations; see also ~\cite{hoehn_how_2018,vanrietvelde_change_2019,hohn_how_2020,vanrietvelde_change_2020,vanrietvelde_switching_2023,Chataignier:2020}) to account for quantum reference frame interactions. In the non-interacting case when the gauge group is a one-dimensional Abelian group, this change-of-frame map reduces to the change-of-frame map introduced in Ref.~\cite{giacominiQuantumMechanicsCovariance2019}, as shown in Ref.~\cite{hohntrinity2021}. In Sec.~\ref{Sec:Applications} we apply the change-of-frame map to temporal and spatial reference frames in the presence of interactions.

In addition, system observables relative to the frame $R$ can be transformed to another frame $\tilde{R}$ with the change-of-frame map. Consider an observable $A$ of $S$ described by the set of effect densities $\{E_S(a) \in \mathcal{E}(\mathcal{H}_{S|R}) \ | \ a \in \Omega_A\}$ relative to $R$ in the orientation $r \in G_R$. The change-of-frame map transforms the effect densities defining $A$ to effect densities defining an observable $\tilde{A}$ of $\tilde{S}$ relative to $\tilde{R}$ in the orientation $\tilde{r} \in G_R$ given by
\eq{\label{eq:AtildeA}
E_{\tilde{S}}(a) =  U(\tilde{r},r) E_{\tilde{S}}(a) U(\tilde{r},r)^\dagger.
}
It follows that the statistics of $\tilde{A}$ when evaluated with $\ket{\psi_{\tilde{S}|\tilde{R}}(\tilde{r})} \in \mathcal{H}_{\tilde{S}|\tilde{R}}$ are the same as the statistics of $A$ when evaluated with $\ket{\psi_{S|R}(r)} \in \mathcal{H}_{\tilde{S}|\tilde{R}}$, where $\ket{\psi_{\tilde{S}|\tilde{R}}(\tilde{r})} = U(\tilde{r},r) \ket{\psi_{S|R}(r)}$. Moreover, the algebra of observables of $\tilde{S}$ relative to $\tilde{R}$ is isomorphic to the algebra of observables of $S$ relative to $R$, that is, $\mathcal{E}(\mathcal{H}_{S|R}) \simeq \mathcal{E}(\mathcal{H}_{\tilde{S}|\tilde{R}})$.

\section{Relational Dirac observables}
\label{Sec:Relational Dirac observables}
The relational description described in Sec.~\ref{Sec:InteractingQRF} can be formulated entirely in terms of the physical Hilbert space. That is,  the statistics of measurement outcomes determined from the conditional state $\ket{\psi_S(r)} \in \mathcal{H}_{S|R}$ of $S$ relative to $R$  can be lifted to the physical Hilbert space and expressed in terms of the expectation values of gauge invariant observables, the so-called relational Dirac observables~\cite{Tambornino:2011,hohntrinity2021,hohn_equivalence_2021,delahametteperspectiveneutral2021}. Gauge invariants are also called complete observables, and throughout the literature they have also been referred to as ``evolving constants of motion'' and ``perennials''~\cite{RovelliECM1,RovelliECM2,rovelli_partial_2002,kuchar_time_2011,dittrich_partial_2006,dittrich_partial_2007,Tambornino:2011,thiemann_modern_2008}.

The following theorem establishes that the statistics of a system observable and a conditional state defined relative to a choice of reference frame and orientation are equivalent to the statistics of induced effect operators on the physical Hilbert space, which are relational Dirac observables.
\begin{theorem}\label{theo:relobs}
    The probability that a system observable $A$, with outcome space $\Omega_A$, event space $\mathcal{F}_A$, and effect operators $\{E_S(a) \in \mathcal{E}(\mathcal{H}_{S|R}) \ | \ a\in \mathcal{F}_A\}$, takes the value $a$ given that the frame has orientation $r\in G_R$ is defined as
  \eq{
        P(a \,{\rm when}\, r) &\ce \bra{\psi_{S}(r)} E_S(a) \ket{\psi_S(r)} \\
&= \bra{\Psi}F_A(r;a) \ket{\Psi}_{\rm phys},
\label{eq:onetimeprobability}
}
where
\begin{equation}
    \begin{aligned}
        F_A(r;a) &\ce \mathcal{R}^{-1}(r)E_S(a)\mathcal{R}(r) \\
        &=\mathcal{G}\left[\Delta_R^{\frac{1}{2}}(r) E_R(r)\otimes   E_S(a)  \Delta_R^{\frac{1}{2}}(r)\right],
        \label{eq:relationalDiracObservable}
    \end{aligned}
\end{equation}
is an effect operator on the physical Hilbert space associated with a gauge-invariant observable, $F_A(r;a) \in \mathcal{E}(\mathcal{H}_{\rm phys})$, which defines a relational Dirac observable that characterizes the statistics of the system observable $A$ relative to a reference frame $R$ in the orientation $r \in G_R$. 
\end{theorem}
\begin{proof}
We begin by establishing the second equality in Eq.~\eqref{eq:relationalDiracObservable} by observing that 
\begin{align}
    F_A(r;a)&\ce\mathcal{R}^{-1}(r)E_S(a)\mathcal{R}(r) \notag\\
    &= \int_G \D g \, U(g) \Delta_R^{\frac{1}{2}}(r)  \ket{r}\!\bra{r} \otimes E_S(a) \Delta_R^{\frac{1}{2}}(r)\notag\\
    &= \int_G \D g \, U(g) \Delta_R^{\frac{1}{2}}(r)  E_R(r) \otimes E_S(a) \Delta_R^{\frac{1}{2}}(r) U(g)^\dagger \notag\\
    &=\mathcal{G}\left[ \Delta_R^{\frac{1}{2}}(r)  E_R(r) \otimes E_S(a) \Delta_R^{\frac{1}{2}}(r)\right].
\end{align}
Then, by direct computation, the second equality in Eq.~\eqref{eq:onetimeprobability} is established:
\begin{align}
        P(a \, {\rm when}\,r)  
        &= \bra{\psi_{S}(r)} E_S(a) \ket{\psi_S(r)} \notag\\
        &= \bra{\Psi}\mathcal{R}(r)^\dagger  E_S(a)\mathcal{R}(r)\ket{\Psi}_{\rm kin}\notag\\
        &= \bra{\Psi} \Delta_R^{\frac{1}{2}}(r) \ket{r}\!\bra{r}\otimes   E_S(a)  \Delta_R^{\frac{1}{2}}(r)\ket{\Psi}_{\rm kin} \notag\\
&=\bra{\psi}\mathcal{G}\left[\Delta_R^{\frac{1}{2}}(r) E_R(r)\otimes   E_S(a)  \Delta_R^{\frac{1}{2}}(r)\right]\ket{\Psi}_{\rm kin}\notag\\
&= \bra{\Psi}F_A(r;a) \ket{\Psi}_{\rm phys}.
\end{align}

Given that $F_A(r;a)$ is defined as the $G$-twirl of $\Delta_R^{\frac{1}{2}}(r) E_R(r) \otimes   E_S(a)  \Delta_R^{\frac{1}{2}}(r)$, it immediately follows that $[F_A(r;a),U(g)] = 0$, ensuring that it is a gauge-invariant operator acting on $\mathcal{H}_{\rm phys}$. Furthermore, due to the linearity and complete positivity of the $G$-twirl and that $\{ E_S(a) \in \mathcal{E}(\mathcal{H}_{S|R}) \ | \ a \in \mathcal{F}_A\}$ is a POVM on $\mathcal{H}_{S|R}$, it also follows that $\{F_A(r;a) \in \mathcal{E}(\mathcal{H}_{\rm phys}) \ | \ a\in \mathcal{F}_A\}$ is a POVM on $\mathcal{H}_{\rm phys}$, where the effect operator $F_A(r;a)$ is associated with outcome $a\in \Omega_A$ when $R$ is in the orientation $r\in G_R$.
\end{proof}

\begin{figure}
    \centering   \includegraphics[width=\linewidth]{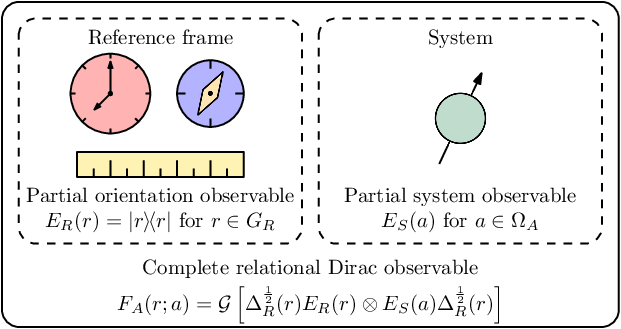}
    \caption{A pictorial representation of a complete, gauge-invariant relational Dirac observable $F_R(r;a)$ that describes a partial system observable with effect operators $E_S(a)$ for $a \in \Omega_A$ relative to a reference frame with partial orientation observable with effect densities $E_R(r) = \ket{r}\!\bra{r}$ for $r\in G_R$.}
    \label{fig:RelationalObservable}
\end{figure}
We see that, in general, $E_R(r)$ and $E_S(a)$ do not lead by themselves to meaningful probability distributions, as they are not gauge invariant. Theorem~\ref{theo:relobs} establishes that it is the gauge-invariant component of the product $\Delta_R^{\frac12}(r)E_R(r)\otimes E_S(a)\Delta_R^{\frac12}(r)$ that has a physical and relational meaning, which can be understood in terms of $E_S(a)$ and the conditional states relative to the chosen frame and orientation. In this way, the gauge-invariant theory allows for probabilistic predictions about suitably defined correlations between the reference frame and system observables, not of these observables by themselves. For this reason, we can refer to $E_R(r)$ and $E_S(a)$ as `partial observables' (their relations are observable), while $F_A(r;a)$ is a `complete observable,' which can be measured and predicted (see Fig.~\ref{fig:RelationalObservable}). The notion of partial and complete observables has been discussed at length in the literature \cite{rovelli_partial_2002,dittrich_partial_2006,dittrich_partial_2007, rovelli_why_2014,thiemann_modern_2008}.

The significance of Theorem~\ref{theo:relobs} is that it establishes the equivalence of the Page--Wootters formalism (i.e., conditional states) and relational Dirac observables in the presence of interactions between the system and quantum reference frame. This extends the previously established equivalence of these formalisms in the non-interacting case~\cite{hohntrinity2021,hohn_equivalence_2021,delahametteperspectiveneutral2021}. In addition, Eq.~\eqref{eq:relationalDiracObservable} provides a generalized relativization map of the partial system observable $E_S(a)$ relative to the partial orientation observable $E_R(r) = \ket{r}\!\bra{r}$, producing a complete relational Dirac observable $F_A(r;a)$. This relativization map reduces to the \$-map~\cite{bartlett_reference_2007},  \textyen-map~\cite{loveridge_symmetry_2018,glowacki_quantum_2023}, and group quantization~\cite{hohntrinity2021,hohn_equivalence_2021,delahametteperspectiveneutral2021} of the series expression of classical relational Dirac observables~\cite{dittrich_partial_2006,dittrich_partial_2007,thiemann_modern_2008,hohntrinity2021} in the non-interacting case (i.e., $\Delta_R(r) = I_R$). A previous version of our definition of quantum relational observables (constructed with a Faddeev--Popov determinant operator) was discussed in \cite{Chataignier:2019,Chataignier:2020} in the context of quantum gravity and cosmology. In \cite{Chataignier:2020}, a change of quantum reference frames was described as a change of basis in the physical Hilbert space. 

It is worthwhile to note that the definition of relational observables via the $G$-twirl given in Eq.~\eref{relationalDiracObservable} (see also \cite{Chataignier:2019,Chataignier:2020}) can be seen as the quantum analogue of the classical integral expression
\eq{
F[f|\chi] = \boldsymbol{\Delta}_{\rm FP}[A]\int\mathcal{D}\varepsilon\,\delta\!\left(\chi[A^\varepsilon]\right)f[A],
}
which produces a gauge-invariant extension of the functional $f[A]$ (partial observable) in a general gauge theory. If $f[A]\equiv1$, this integral formula reduces to the Faddeev--Popov resolution of the identity given in Eq.~\eref{FP}.

The following theorem establishes that the relativization map in Eq.~\eqref{eq:relationalDiracObservable} is a homomorphism between observables defined on $\mathcal{H}_{S|R}$ and relational Dirac observables defined $\mathcal{H}_{\rm phys}$, which generalizes a similar result previously established in the absence of frame-system interactions (i.e., tensor product representations of the gauge group)~\cite{hohntrinity2021,loveridge_symmetry_2018, delahametteperspectiveneutral2021}.
\begin{theorem}\label{theo:homomorphism}
    The map 
    \begin{equation}
        \begin{aligned}
            F_A(r) : \mathcal{L}(\mathcal{H}_{S|R}) &\to \mathcal{L}(\mathcal{H}_{\rm phys}), \\
            A & \mapsto F_A(r) = \mathcal{R}^{-1}(r) A \mathcal{R}(r),
        \end{aligned}
    \end{equation}
    which encodes the algebra of observables of $S$ relative to $R$ at the orientation $r$,  is a $*$-algebra homomorphism. That is, 
    $F_{AB+C}(r) = F_{A}(r) F_{B}(r) +F_{C}(r)$
    and
    $F_{A^\dagger}(r) = F_A(r)^\dagger$, for $A,B,C \in \mathcal{L}(\mathcal{H}_{S|R})$.
\end{theorem}
\begin{proof}
    Observe that
\eq{
              F_{AB+C}(r) &=  \mathcal{R}^{-1}(r) (AB+C) \mathcal{R}(r) \\
              &=  \mathcal{R}^{-1}(r) A\mathcal{R}(r)\circ \mathcal{R}^{-1}(r)B\mathcal{R}(r)\\
              &\ +\mathcal{R}^{-1}(r)C \mathcal{R}(r) \\
               &=F_{A}(r) F_{B}(r) +F_{C}(r),
}
    where the linearity of the reduction map and its inverse was used and Eq.~\eqref{eq:rightInverse} was used to arrive at the second equality. 

    Moreover,
    \begin{equation}
        \begin{aligned}
            F_{A^\dagger}(r) &= \mathcal{R}^{-1}(r) A^\dagger \mathcal{R}(r) \\
            &= \mathcal{G}\left[\Delta_R^{\frac{1}{2}}(r) \ket{r}\!\bra{r}\otimes   A^\dagger  \Delta_R^{\frac{1}{2}}(r)\right] \\
            &= \mathcal{G}\left[\Delta_R^{\frac{1}{2}}(r) \ket{r}\!\bra{r}\otimes   A  \Delta_R^{\frac{1}{2}}(r)\right]^\dagger \\
            &=F_{A}(r)^\dagger.
        \end{aligned}
    \end{equation}
    
    The self-adjoint operators $A,B,C \in \mathcal{L}(\mathcal{H_{S|R}})$ can be taken to represent sharp observables. However, the homomorphism extends to effect operators describing unsharp observables characterized by POVMs because $\mathcal{E}(\mathcal{H_{S|R}}) \subset\mathcal{L}(\mathcal{H_{S|R}})$. That is, the map $F_A(r;a) \ce \mathcal{R}^{-1}(r) E_S(a)\mathcal{R}(r)$ is also a $*$-algebra homomorphism.
\end{proof} 

Notice that, given the transformation of system observables between frames $R$ and $\tilde{R}$ in Eq.~\eref{AtildeA}, we can lift the effect operators $E_S(a)$ and $E_{\tilde{S}}(a)$ associated with the observables $A$ and $\tilde{A}$ to effect operators $F_A(r;a)$ and $F_{\tilde{A}}(\tilde{r};a)$ on the physical Hilbert space. Using Eqs.~\eref{Iphys},~\eref{change-frame-map} and~\eref{Udagger} together with Eq.~\eref{AtildeA}, we find the identity
\begin{align}
F_{\tilde{A}}(\tilde{r};a) &= \mathcal{\tilde{R}}^{-1}(\tilde{r})E_{\tilde{S}}(a)\mathcal{\tilde{R}}(\tilde{r})\notag\\
&= \mathcal{\tilde{R}}^{-1}(\tilde{r})U(\tilde{r},r)E_{S}(a)U(\tilde{r},r)^{\dagger}\mathcal{\tilde{R}}(\tilde{r})\notag\\
&= \left[\mathcal{\tilde{R}}^{-1}(\tilde{r})\circ\mathcal{\tilde{R}}(\tilde{r})\right]F_{A}(r;a)\left[\mathcal{\tilde{R}}^{-1}(\tilde{r})\circ\mathcal{\tilde{R}}(\tilde{r})\right]\notag\\
&= F_{A}(r;a).\label{eq:relobs-change-frame}
\end{align}
We thus see that the relational Dirac observables $F_{\tilde{A}}(\tilde{r};a)$ and $F_{A}(r;a)$ are equal, meaning they encode the same gauge-invariant information associated with the correlation between a system observable, represented either by $A$ in the reference frame $R$ or by $\tilde{A}$ in the reference frame $\tilde{R}$, and a frame orientation, represented by $r$ in $R$ or $\tilde{r}$ in $\tilde{R}$. Indeed, the statistics are the same regardless of which frame is chosen in the reduction to conditional states, since Eqs.~\eref{onetimeprobability} and~\eref{relobs-change-frame} together imply
\eq{
P(a \,{\rm when}\, r) &\ce \bra{\psi_{S}(r)} E_S(a) \ket{\psi_S(r)} \\
&= \bra{\Psi}F_A(r;a) \ket{\Psi}_{\rm phys}\\
&= \bra{\Psi}F_{\tilde{A}}(\tilde{r};a) \ket{\Psi}_{\rm phys}\\
&= \bra{\psi_{\tilde{S}}(\tilde{r})} E_{\tilde{S}}(a) \ket{\psi_{\tilde{S}}(\tilde{r})}\\
&=P(a \,{\rm when}\, \tilde{r}).
}

\section{Applications}
\label{Sec:Applications}

\subsection{Temporal frames}
\label{Sec:TemporalFrames}
We now apply the formalism developed in Secs.~\ref{Sec:InteractingQRF} and \ref{Sec:Relational Dirac observables} to temporal reference frames (i.e., clocks) by supposing that the gauge group $G$ is a one-dimensional Lie group. Despite this simplification, such gauge groups are sufficient to describe the dynamics of standard Hamiltonian systems~\cite{dirac_lectures_1964, rovelli_quantum_2004}. In particular, if a theory does not single out a preferred clock, as is the case in general relativity, then it ought to be invariant under relabeling of the employed time variable, which corresponds to invariance under the one-dimensional Lie group $G$ generated by the Hamiltonian $H$ of the theory. By construction, the kinematical Hilbert space $\mathcal{H}_{\rm kin}$ carries a unitary representation of $G$, denoted as $U(g) = e^{-igH}$ for all $g\in G$. The physical Hilbert space $\mathcal{H}_{\rm phys}$ is the set of states that are invariant under the unitary action of $G$, $U(g)\ket{\Psi} = \ket{\Psi}$, equipped with the physical inner product given in Eq.~\eref{PhysInnerProduct}. Differentiating $\ket{\Psi} = U(g)\ket{\Psi} = \e^{-i g H}\ket{\Psi}$ with respect to $g$, one finds the constraint equation
\eq{\label{eq:WDW}
H\ket{\Psi} = 0 \ ,
}
and from $U(g)\int_G\D g'\,U(g') = \int_G\D g''\, U(g'')$, one similarly finds 
\eq{\label{eq:WDW-projector}
H\int_G\D g\, U(g) = 0 \ .
}
We partition the kinematical Hilbert space as $\mathcal{H}_{\rm kin} \simeq \mathcal{H}_C\otimes\mathcal{H}_S$, where $\mathcal{H}_C$ describes a clock and $\mathcal{H}_S$ is a system containing all other degrees of freedom. Without loss of generality, we take the Hamiltonian to be
\begin{equation}
    H \ce H_C \otimes I_S + I_C \otimes H_S + H_{\rm int},
    \label{eq:totalHamiltonian}
\end{equation}
where $H_C$, $H_S$, and $H_{\rm int}$ are self-adjoint operators that correspond to the clock, system, and interaction Hamiltonians, respectively.\footnote{Notice that, if $H_{\rm int}=0$, and $H_C$ and $H_S$ generate representations of $G$ then $U(g)$ is a tensor product representation of $G$; however, in what follows, we need not assume this is the case.} 

The group of temporal reference frames is generated by the clock Hamiltonian $H_C$ and represented on the clock factor of the kinematical Hilbert space as $U_C : G_C \to \mathcal{B}(\mathcal{H}_C),\ t \mapsto U_C(t) = e^{-itH_C}$; depending on the spectral properties of $H_C$, $G_C$ is either $S^1$ or $\mathbb{R}$. The time variable $t$ labels different orientations (i.e., times) of the clock. The orientation observable is known as a time observable~~\cite{holevo_probabilistic_1982,busch_quantum_2016,smith_time_2025}, and is covariant with respect to $G_C$ and its effect density is given by $\{E_C(t) = \ket{t}\!\bra{t} \ | \ t\in G_C\}$, where $\ket{t} = U_C(t)\ket{t_0}$ are group coherent states~\cite{perelomov_generalized_1986,zhang_coherent_1990}.

The conditional state of $S$ when the clock reads the time $t$ is given by Eq.~\eqref{eq:conditionalState}:
\eq{
\ket{\psi_S(t)} \ce \mathcal{R}(t)\ket{\Psi} = \bra{t} \otimes I_S \Delta^{\frac{1}{2}}(t) \ket{\Psi},
\label{eq:timedependentConditionalState}
}
where $\Delta(t)$ is the inverse of $\mathcal{G}[E_C(t) \otimes I_S]$, which was denoted as $\Delta_R(r)$ above, and $\ket{\Psi} \in \mathcal{H}_{\rm phys}$. Time evolution of the conditional state of $S$ relative to the clock $C$ is obtained by applying the change-of-frame map to change the orientation of the time variable from $t_1$ to $t_2$
\eq{
\ket{\psi_S(t_2)} = U_{S|C}(t_2,t_1) \ket{\psi_S(t_1)},
\label{eq:TimeEvolutionOperator}
}
where $U_{S|C}(t_2,t_1) \ce \mathcal{R}(t_2) \circ \mathcal{R}^{-1}(t_1)$ which can be verified by direct computation similar to the proof in the corollary following Theorem~\ref{thm:two}. Accordingly, we will refer to $U_{S|C}(t_2,t_1)$ as the time evolution operator, which is a special case of Eq.~\eqref{eq:OrientationChange}. As shown in Appendix~\ref{SchrodingerDerivation}, the time evolution operator can be written as the time-ordered exponential
\eq{
U_{S|C}(t_2,t_1) = \mathcal{T} \exp \left[-i\int_{t_1}^{t_2}dt \, H_{S|C}(t) \right],
\label{eq:timeorderedExponential}
}
where
\eq{
H_{S|C}(t) &\ce H_S \\
&\quad +\! \int_G \!\D g  \bra{t} \!\left[ H_{\rm int} \Delta(t) \!+\! i \frac{d\Delta^{\frac{1}{2}}(t)}{dt} \Delta^{\frac{1}{2}}(t)\right]\! U(g)\!\ket{t}\!,
\label{eq:HSC}
}
generates the time evolution of $S$ relative to $C$, and is thus interpreted as the Hamiltonian of $S$ relative to $C$. Equivalently, by differentiating Eq.~\eqref{eq:TimeEvolutionOperator} and using Eq.~\eqref{eq:timeorderedExponential}, the conditional state is seen to satisfy the Schr\"{o}dinger equation
\eq{
     i\frac{d}{dt} \ket{\psi_S(t)} = H_{S|C}(t) \ket{\psi_S(t)}.   \label{eq:SchrodingerEquation}
}
The Hamiltonian $H_{S|C}(t)$ has two terms: the first is simply the kinematical system Hamiltonian $H_S$, and the second is an integral operator that is $t$-dependent and gauge-nonlocal. On account of Theorem~\ref{thm:two}, the evolution described by the Schr\"{o}dinger equation is necessarily unitary for all $H_{\rm int}$ for which $\Delta(t)$ can be constructed. Let us now examine some important particular cases to illustrate this result.

The standard Schr\"odinger equation generated by $H_S$ is recovered when $C$ and $S$ do not interact, $H_{\rm int} = 0$, and the clock evolution and gauge groups coincide, $G_C = G$.\footnote{A discussion of the $G_C = S^1\neq G$ case for mechanical models with $H_{\rm int} = 0$ can be found in \cite{chataignier_relational_2026}.} Then, $U(g) = U_C(g) \otimes U_S(g)$, where $U_S(g) = e^{-igH_S}$. The  $G$-twirl of the orientation effect density $E_C(t)=\ket{t}\!\bra{t}$ reads
\begin{equation}
    \mathcal{G}[E_C(t)\otimes I_S] = \int_G dg \, U_C(g) E_C(t) U_C(g)^\dagger \otimes I_S = I_{\rm kin},
\end{equation}
and its inverse is trivial, $\Delta(t) = I_{\rm kin}$. This results in the second term in Eq.~\eqref{eq:HSC} vanishing and the usual Schr\"{o}dinger equation generated by $H_S$ is obtained. It is in this case that the dynamics of orthodox quantum theory is recovered, and the reduction map in Eq.~\eqref{eq:conditionalState} reduces to the standard Page--Wootters reduction map, 
\eq{
    \mathcal{R}(t) = \bra{t} \otimes I_S.
    \label{eq:PWreduction}
}

\subsubsection{Time-nonlocal Schr\"odinger equation}
Now consider the case in which $H_{\rm int}\neq0$. In \cite{smith_quantizing_2019}, a general interaction $H_{\rm int}\neq0$ was considered and a time-nonlocal Schr\"odinger equation,
\eq{\label{eq:nonlocal}
i\frac{d}{dt} \ket{\psi_S(t)} &= H_S \ket{\psi_S(t)}\\
&\quad +\int_{G_C}\D t'\,\bra{t}H_{\rm int}\ket{t'}\ket{\psi_S(t')},
}
was obtained for a conditional state defined with the standard Page--Wootters reduction map in Eq.~\eqref{eq:PWreduction}, $\ket{\psi_S(t)}\ce\bra{t}\otimes I_S\ket{\Psi}$. Clearly, this conditional state only agrees with the definition in Eq.~\eqref{eq:conditionalState} if $\Delta(t) =  I_{\rm kin}$. Otherwise, this standard reduction map will generally not lead to unitary evolution. As was noted in \cite{smith_quantizing_2019}, Eq.~\eref{nonlocal} does not describe unitary evolution for general interactions between $C$ and $S$. Here, we see that to ensure unitarity one must adopt the definition of the conditional state given in Eq.~\eqref{eq:conditionalState}, which leads to Eq.~\eqref{eq:SchrodingerEquation}. Not including the factor of $\Delta^{\frac12}(t)$ in Eq.~\eqref{eq:conditionalState} is the analogue of neglecting the proper gauge-fixing measure in gauge theories, which leads to well-known problems with unitarity~\cite{Feynman63,HT:book,Teitelboim1991,Barvinsky_1993}.

However, it is worthwhile to examine how the time-nonlocal Schr\"{o}dinger equation in Eq.~\eqref{eq:nonlocal} can be obtained as a particular case of Eq.~\eqref{eq:SchrodingerEquation}. First, we assume that $H_{\rm int}~\neq~0$, but $\Delta(t)\equiv\Delta$ is independent of $t$. In this case, Eqs.~\eqref{eq:SchrodingerEquation} and~\eqref{eq:HSC} lead to
\eq{\label{eq:pre-nonlocal}
i\frac{d}{dt} \ket{\psi_S(t)} &= H_S \ket{\psi_S(t)}\\
&\quad +\bra{t}H_{\rm int}\Delta\int_{G}\D g\,U(g)\ket{t}\!\bra{t}\otimes I_S\Delta^{\frac12}\ket{\Psi}.
}
Using the invariance of physical states, $\ket{\Psi} = U^{\dagger}(g)\ket{\Psi}$ and the fact that $\Delta$ commutes with $U(g)$, the last term in Eq.~\eref{pre-nonlocal} becomes
\eq{\notag
&\bra{t}H_{\rm int}\Delta\mathcal{G}[E_C(t)\otimes I_S]\Delta^{\frac12}\ket{\Psi}\\ 
&\hspace{10em} = \bra{t}H_{\rm int}\Delta^{\frac12}\ket{\Psi}\\
&\hspace{10em} = \int_{G_C}\D t'\bra{t}H_{\rm int}\ket{t'}\ket{\psi_S(t')},
}
where the completeness of the clock effect densities, $\int_{G_C} \D t \,\ket{t}\!\bra{t} = I_C$, was used together with Eq.~\eref{timedependentConditionalState} in going from the second to the third equality. Thus, Eq.~\eref{pre-nonlocal} becomes Eq.~\eref{nonlocal}.

\subsubsection{Weakly interacting clock and system}
If $G = G_C$ and the strength of the interaction between the clock and system is governed by a small parameter, $H_{\rm int} = \lambda \bar{H}_{\rm int}$ with $\lambda\ll1$, we can use perturbation theory. At lowest order in $\lambda$, we have $U_0(g) \ce U(g)|_{\lambda = 0} = U_C(g)\otimes U_S(g)$ and the $G$-twirl reads
\begin{align}
\mathcal{G}[E_C(t)\otimes I_S] &= \int_G \D g\, U_0(g)[E_C(t)\otimes I_S] U_0(g)^{\dagger}+\mathcal{O}(\lambda)\notag\\
&= I_{\rm kin} +\mathcal{O}(\lambda) .\label{eq:pert-G}
\end{align}
This also implies the identity
\begin{align}
 \int_G \D g\, U_0(g)E_C(t)\otimes I_S\ket{\Psi} &= \int_G\D g\,U_0(g)\ket{t}\ket{\psi_S(t)}\notag\\
 &= \ket{\Psi}+\mathcal{O}(\lambda) .\label{eq:pert-UE}
\end{align}
It is clear from Eq.~\eref{pert-G} that the $G$-twirl is positive definite and its inverse $\Delta(t)$ can be computed perturbatively, being $t$-independent at lowest order, $\Delta(t) = I_{\rm kin}+\mathcal{O}(\lambda)$, so $d \Delta(t)/ d t = \mathcal{O}(\lambda)$. With this, the Schr\"odinger equation~\eref{SchrodingerEquation} becomes
\eq{\label{eq:pert-schro}
\!\!i\frac{d}{dt} \ket{\psi_S(t)} &= H_{S} \ket{\psi_S(t)}\\
&\quad +\bra{t}\!\left(\lambda\bar{H}_{\rm int}+\frac{i}{2}\frac{d\Delta(t)}{d t}\right)\!\ket{\Psi} + \mathcal{O}(\lambda^2),
}
where we used the identity in Eq.~\eref{pert-UE}. From
\eq{
     \frac{d \Delta(t)}{dt}
     &= \frac{d}{dt} \mathcal{G}\left[ E_C(t) \otimes I_S\right]^{-1} \\
    &=- \Delta(t) \left(\frac{d}{dt} \mathcal{G}\left[ E_C(t) \otimes I_S\right] \right) \Delta(t)  \\
    &=-i \Delta(t)  \mathcal{G}\left[ \left[E_C(t),H_C\right] \otimes I_S\right]  \Delta(t),
}
together with $[\Delta(t),H] = 0$ and Eq.~\eref{WDW-projector}, we find
\eq{\label{eq:pert-delta-der}
\!\frac{i}{2}\frac{d\Delta(t)}{d t}\ket{\Psi} = \frac{\lambda}{2}\int_G\D gU_0(g)[\bar{H}_{\rm int},E_C(t)\otimes I_S]\ket{\Psi},
}
neglecting terms of order $\lambda^2$. Then, Eq.~\eref{pert-schro} can be written as the time-nonlocal equation
\eq{\label{eq:pert-nonlocal}
\!i\frac{d}{dt} \ket{\psi_S(t)} &= H_{S} \ket{\psi_S(t)}\\
&\quad +\lambda\int_{G_C}\D t'\braket{t|K|t'}\ket{\psi_S(t')}+\mathcal{O}(\lambda^2),
}
where 
\eq{\label{eq:pert-K}
\!\!K\ce \bar{H}_{\rm int}+\frac12\int_G\D gU_0(g)[\bar{H}_{\rm int},E_C(t)\otimes I_S]U_0(g)^{\dagger}.
}
Equation~\eref{pert-nonlocal} can be seen as a perturbatively corrected version of the previously discussed time-nonlocal Schr\"{o}dinger equation~\eref{nonlocal}, for which $\Delta$ was independent of time.

By construction, the dynamics of $\ket{\psi_S(t)}$ is unitary and the inner product between two conditional states is conserved under time evolution. Indeed, using Eqs.~\eref{pert-nonlocal} and~\eref{pert-K} for a pair of solutions $\ket{\psi_S(t)} = (\bra{t}\otimes I_S)\ket{\Psi}+\mathcal{O}(\lambda)$ and $\ket{\phi_S(t)} = (\bra{t}\otimes I_S)\ket{\Phi}+\mathcal{O}(\lambda)$ with $\ket{\Psi},\ket{\Phi}\in\mathcal{H}_{\rm phys}$, we find
\eq{\notag
\frac{d}{dt}\braket{\phi_S(t)|\psi_S(t)} &= i\lambda\bra{\Phi}K^{\dagger}(E_C(t)\otimes I_S)\ket{\Psi}\\
&\quad -i\lambda\bra{\Phi}(E_C(t)\otimes I_S) K\ket{\Psi}+\mathcal{O}(\lambda^2)\\
&= i\lambda\!\bra{\Phi}\![\bar{H}_{\rm int},E_C(t)\otimes I_S]\!\ket{\Psi}\\
&\quad+i\lambda\!\bra{\Phi}\![E_C(t)\otimes I_S,\bar{H}_{\rm int}]\!\ket{\Psi}+\mathcal{O}(\lambda^2)\\
&= \mathcal{O}(\lambda^2) ,
}
where we also used Eq.~\eref{pert-UE}.

An interesting case is that in which $[H_C\otimes I_S,H_{\rm int}] = [I_C\otimes H_S,H_{\rm int}] = 0$, so the clock, system, and interaction Hamiltonians commute with one another. Then Eq.~\eref{pert-K} simplifies to
\eq{
\!\!K &= \bar{H}_{\rm int}+\frac12\int_{G_C}\D gU_0(g)[\bar{H}_{\rm int},E_C(t)\otimes I_S]U_0(g)^{\dagger}\\
&= \bar{H}_{\rm int}+[\bar{H}_{\rm int},I_{\rm kin}]\\
&=\bar{H}_{\rm int},
}
given that $G = G_C$. With this, Eq.~\eref{pert-nonlocal} coincides with the time-nonlocal Schr\"{o}dinger equation~\eref{nonlocal} up to terms of order $\lambda^2$. This is because $\Delta(t)$ is independent of $t$ in this case. Indeed, we find
\eq{\notag
\Delta^{-1}(t)&\ce\mathcal{G}[\ket{t}\!\bra{t} \otimes I_S]\\
&= \int_G \D g \, U(g)\ket{t}\!\bra{t} \otimes I_S U(g)^\dagger \\ 
&= \int_{G_C} \D g \, e^{-igH_{\rm int}}\ket{t+g}\!\bra{t+g} \otimes I_S e^{igH_{\rm int}} \\ 
&= I_{\rm kin}+i\lambda[T\otimes I_S,\bar{H}_{\rm int}]+\mathcal{O}(\lambda^2),
}
where
\eq{
T = \int_{G_C}\mathrm{d}t\,t\ket{t}\!\bra{t}
}
is the first moment of the clock POVM. We thus see that $\Delta^{-1}(t)$ is independent of $t$ at this order of approximation with the $t$-independent inverse
\eq{
\Delta = I_{\rm kin}-i\lambda[T\otimes I_S,\bar{H}_{\rm int}]+\mathcal{O}(\lambda^2).
}
Precisely this kind of conditional dynamics is obtained for gravitationally interacting clocks and systems~\cite{smith_quantizing_2019,castro-ruiz_quantum_2020,cafasso_quantum_2024}. The interaction Hamiltonian is of the form $H_{\rm int} = -G H_C\otimes H_S/(c^4 d)$, which commutes with $H_C\otimes I_S$ and $I_C\otimes H_S$. The role of $\lambda$ is played by $-G/(c^4d)$, which depends on Newton's constant $G$, the speed of light in the vacuum $c$, and the distance $d$ between the clock and the system. In \cite{smith_quantizing_2019}, it was noted that Eq.~\eref{nonlocal} was unitary for this system. Here, we see that indeed it is, as a particular, perturbative case of Eq.~\eref{SchrodingerEquation}. 

\subsubsection{Two-time conditional probabilities and Kucha\v{r}'s criticism}

One might argue, as Kucha\v{r} did~\cite{kuchar_time_2011} (see also Ref.~\cite{page_clock_1994}), that dynamics should not emerge from the one-time probability of `$a$ when $t$' computed in Eq.~\eqref{eq:onetimeprobability}, which implies that the conditional state in Eq.~\eqref{eq:timedependentConditionalState} satisfies the Schr\"{o}digner equation in Eq.~\eqref{eq:SchrodingerEquation}, but instead from two-time conditional probabilities of the sort `$b$ when $t_2$' given that `$a$ when $t_1$', where $a$ and $b$ are measurement outcomes of partial system observables at the times $t_1$ and $t_2$. We show here, following a similar approach as Ref.~\cite{hohntrinity2021}, that even in the presence of reference frame interactions, these two-time conditional probabilities yield the same dynamics implied by the one-time conditional probabilities and recover the same propagator generated by the Hamiltonian in Eq.~\eqref{eq:SchrodingerEquation}.

Suppose that we are interested in measurement outcomes in $a\in \mathcal{F}_A$ and $b\in\mathcal{F}_B$, which are  associated with the effect operators $E_A(a) \in \mathcal{E}(\mathcal{H}_S)$ and $E_B(b)\in \mathcal{E}(\mathcal{H}_S)$. Then the two-time conditional probability that $b$ when $t_2$ conditioned on $a$ when $t_1$ is given by the standard Born rule in terms of relational Dirac observables as
\eq{
    &P\left(b \mathrm{\ when \ } t_2 \, |\,  a \mathrm{\ when \ } t_1\right) \\
    &\hspace{40pt}\ce \frac{\bra{\Psi} F_A(t_1;a)F_B(t_2;b) F_A(t_1;a)\ket{\Psi}}{\bra{\Psi}F_A(t_1;a) \ket{\Psi}}, 
    \label{eq:two-timeConditionalProbability}
}
where $F_A(t_1;a) \ce \mathcal{R}^{-1}(t_1) E_A(a) \mathcal{R}(t_1)$ and $F_A(t_2;a) \ce \mathcal{R}^{-1}(t_2) E_B(b) \mathcal{R}(t_2)$. Supposing that the partial system observables are sharp, so that the effect operators are rank-1 projectors, $E_A(a) =\ket{a}\!\bra{a}$ and $E_B(b) =\ket{b}\!\bra{b}$, then by the definition of the time evolution map $U_{S|C}(t_2,t_1) \ce \mathcal{R}(t_2) \circ \mathcal{R}^{-1}(t_1)$ and direct computation it follows that
\begin{equation}
    P\left(b \mathrm{\ when \ } t_2 \, |\,  a \mathrm{\ when \ } t_1\right) = \abs{\bra{b}  U_{S|C}(t_2,t_1)  \ket{a}}^2.
\end{equation}
That is, the dynamics that emerges from the conditional probability in Eq.~\eqref{eq:two-timeConditionalProbability} is consistent with the time-evolution map propagating the conditional state $\ket{\psi_S(t_1)} = \ket{a}$ from $t_1$ to $t_2$, and computing the probability that at $t_2$ the conditional state is $\ket{\psi_S(t_2)} = \ket{b}$. 

The significance of this result is that, even in the presence of clock-system interactions, the relational formalism discussed here is agnostic about how dynamics is formulated, either in terms of one-time conditional probabilities and a conditional state governed by the Schrodinger equation that `collapses' after a measurement, or two-time conditional probabilities from which the same dynamics emerges from the Born rule alone. The latter is formulated in terms of a physical state that is time-independent and used to compute multi-time probabilities, and, in contrast to the conditional state, does not `collapse' after a measurement.

In this way, the distinction between two different evolutions\,---\,unitary evolution governed by the Schr\"{o}dinger equation and measurement-induced collapse\,---\,disappears, which is often invoked as a manifestation of the measurement problem. Instead, in the relational formalism discussed here, there is just the Born rule applied at the level of the physical Hilbert space, out of which unitary time evolution emerges relative to an appropriately defined clock  (i.e., temporal reference frame) and is not postulated a priori.

Given the emergence of unitary evolution from a (non-collapsing) physical state, relational Dirac observables, and the Born rule, two possibilities are worth considering. First, one can take the Born rule to be fundamental, and then unitary evolution naturally emerges in the way described above. In a sense, this is the opposite point of view of the usual Everettian interpretation, which is based on a time-dependent Schr\"odinger equation. There, one also has a global quantum state, but the unitary time evolution of this global state is taken as the fundamental ingredient, from which the Born rule and probabilities ought to emerge. 

The second possibility to consider is that neither the Born rule nor unitary evolution should be postulated. One may take the (non-collapsing) global physical state $\ket{\Psi}$ to be a fundamental representation of reality and postulate only its gauge invariance, $U(g)\ket{\Psi} = \ket{\Psi}$, without singling out any clock or time parameter. From this and the structure of the physical Hilbert space, the Born rule may emerge relative to the physical inner product (see \cite{ChataigOrigin} for more details). Then, as a subsequent step, unitary time evolution emerges as described above or as discussed in \cite{ChataigOrigin}. This would amount to a ``timeless Everettian view,'' in which neither Born probabilities nor unitary evolution are basic ingredients.

One can extend the two-time conditional probability in Eq.~\eqref{eq:two-timeConditionalProbability} by further conditioning on additional events associated with other relational Dirac observables to compute multi-time conditional probabilities, or generalize this two-time conditional probability to other reference frames (e.g., outcome $b$ when located at position $y$ conditioned on outcome $a$ when located at position $x$).

\subsection{Changes of spatial frames}
\label{Sec:ChangeOfSpatialFrame}

The change-of-frame map defined in Eq.~\eqref{eq:change-frame-map} can also be used to change between the perspectives of different spatial references. Indeed, a surprising result that emerged in this context is that spatial coherence relative to one frame appears as entanglement relative to another frame~\cite{giacominiQuantumMechanicsCovariance2019,de_la_hamette_entanglement-asymmetry_2022}, which is a consequence of a reference frame-dependent re-factorization of the physical Hilbert space~\cite{ali_ahmad_quantum_2022}. Here, we consider this trade-off when the reference frames interact, and demonstrate that, depending on the situation, for a given amount of coherence in one frame, the amount of entanglement in another reference frame can depend on the strength of interaction between the frames.

Here, we consider three particles, $A$, $B$, and $C$, each described by the Hilbert space $\mathcal{H}_i \simeq L^2(\mathbb{R})$ for $i\in\{A,B,C\}$, so that the kinematical Hilbert space is $\mathcal{H}_{\rm kin} \simeq \mathcal{H}_A\otimes \mathcal{H}_B\otimes \mathcal{H}_C$. In addition, suppose  the gauge group $G$ is generated by the constraint
\eq{
 C &= P_A + P_B + P_C  \\
 &\quad + \lambda \left(  P_A \otimes P_B + 
      P_B \otimes P_C + P_A \otimes P_C \right),
}
where $P_i$ is the momentum operator on $\mathcal{H}_i$ and $\lambda\in\mathbb{R}$ is the strength of the pair-wise interactions between $A$, $B$, and $C$. The group of reference frames for a particle is the group of spatial translations, and a measurement of the particle's position determines the orientation of the reference frame.  Position measurements are described by a covariant projective valued measure characterized by the effect density $\{E_i(x) = \ket{x}\!\bra{x}  \forall x\in\mathbb{R}\}$, where $\ket{x} = \frac{1}{\sqrt{2\pi}} \int \D p_i \, e^{-ixp_i} \ket{p_i}$ are generalized eigenstates of the $i$th particle's position operator and $\ket{p_i}$ are generalized eigenstates of $P_i$; the states $\ket{x}$ form a set of coherent states with respect to spatial translations~\cite{perelomov_generalized_1986,zhang_coherent_1990}, and are associated with the spatial orientation of the $i$th reference frame. Note that in the limit that $\lambda$ vanishes, the gauge group corresponds to a spatial translation of all three particles.

As computed in Appendix~\ref{Changing spatial quantum reference frames}, the change-of-frame map in Eq.~\eqref{eq:change-frame-map} that transforms the description of $BC$ relative to $A$ to the description of $AC$ relative to $B$ takes the form
\eq{
U_{A\to B}(y,x) = \int_G \D g \, \bra{y} \Delta_{B}^{\frac{1}{2}}(y)U(g) \Delta_{A}^{\frac{1}{2}}(x) \ket{x},
}
where we have included the subscript $A\to B$ to emphasize we are changing from frame $A$ to $B$, have suppressed factors of the identity for readability,  $\ket{x}$ and $\ket{y}$ are the orientations of $A$ and $B$, respectively, and 
\eq{
\Delta_A^{\frac{1}{2}} &=  \sqrt{1+ \lambda \left(  P_B + P_C \right) },\\
\Delta_B^{\frac{1}{2}} &=   \sqrt{1+ \lambda \left(  P_A +  P_C\right) },
}
which are seen to be independent of the orientations $x$ and $y$.

\begin{figure}[t!]
    \centering
    \includegraphics[width=\linewidth]{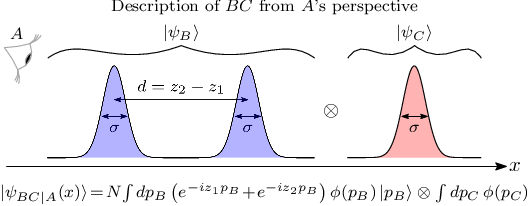}\\ \ \\
    \includegraphics[width=\linewidth]{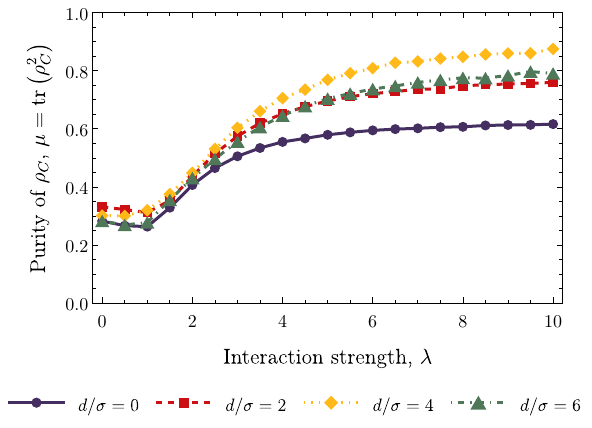}
    
    \caption{(Top) A pictorial representation of the joint state $\ket{\psi_{BC|A}(x)}$ of $BC$ from the perspective of $A$. (Bottom) A plot of the purity $\mu \ce \tr \left(\rho_C^2\right)$ of the reduced state of $C$ from the perspective of $B$ as a measure of entanglement between $A$ and $C$ in the transformed state $\ket{\psi_{AC|B}(y)}$ of $AC$ from the perspective of $B$ for different values of $d \ce z_2 - z_1 / \sigma$. The purity was computed numerically in Mathematica with a precision goal of two decimal places using the method ``MultidimensionalRule''. }
    \label{fig:purity}
\end{figure}

To exhibit the effect of an interaction (i.e., $\lambda \neq 0$), suppose that from the perspective of $A$, the joint state of $BC$ is separable and given by $\ket{\psi_{BC|A}(x)} = \ket{\psi_B} \ket{\psi_C}$. Further suppose that the state of $B$ is a spatial superposition of Gaussian wave packets localized at $z_1$ and $z_2$,
\eq{\ket{\psi_B} 
    &= N \int \D p_B \, \left(e^{-ip_Bz_1}+ e^{-ip_Bz_2} \right) \phi(p_B) \ket{p_B},
}
where $N \ce 1/\sqrt{2 +2e^{-(z_1-z_2)^2 /4\sigma^2}}$ is a normalization constant and $\phi(p) \ce \sqrt{\sigma} e^{- \sigma^2 p^2/2 }/ \pi^{\frac{1}{4}}$ with $\sigma$ being the spread of the Gaussian wave packet in position space, and suppose that the state of $C$ is a Gaussian wave packet localized the origin
\eq{
\ket{\psi_C} = \int \D p_C \, \phi(p_C)\ket{p_C}.
}
From the perspective of $B$, the state of $AC$ is given by
\eq{
\ket{\psi_{AC|B}(y)}&= U_{A\to B} (y,x)\ket{\psi_{BC|A}(x)}\\
&= N \!\int \!\D p_A \D p_C \,\frac{ \sqrt{1+ \lambda p_C + \lambda^2 p_C ^2} }{\abs{1+\lambda \left( p_A +  p_C \right)}} \\
&\quad \times   e^{-ixp_A} \!\left(\! e^{-i \tilde{p}(z_1-y)} \!+\! e^{-i\tilde{p}(z_2-y)}\!\right)\!  \\
&\quad \times \phi(\tilde{p}) \phi(p_C) \ket{p_A}\ket{p_C},
\label{eq:TransformedState}
}
where
\eq{
\tilde{p} \ce - \frac{p_A+p_C + \lambda  p_A p_C }{1+\lambda \left( p_A +  p_C \right) }.
}

It's clear from Eq.~\eqref{eq:TransformedState} that the transformed state $\ket{\psi_{AC|B}(y)}$ describes an entangled state of $A$ and $C$ relative to $B$, despite the state of $A$ and $B$ relative to $C$ being a product state. As a measure of this entanglement, in Fig.~\ref{fig:purity} we plot the purity of the reduced state of $C$, $\mu \ce \tr \left(\rho_C^2\right)$, where $\rho_C \ce \tr_A \left( \ket{\psi_{AC|B}(y)}\! \bra{\psi_{AC|B}(y)} \right)$, as a function of the interaction strength $\lambda$. It is seen that the entanglement between $AC$ depends non trivially on $\lambda$: in all cases, as $\lambda$ increases from $0$ to $1$, entanglement decreases, then as $\lambda$ increases from $1$, entanglement increases monotonically, approaching an asymptotic value. It is also seen that the resulting entanglement between $AC$ depends non-trivially on the scale $d\ce z_2-z_1$ of the superposition of $B$, and is independent of the orientations $x$ and $y$.

\section{Conclusion and Outlook}
\label{Sec:ConclusionAndOutlook}
We have developed a formulation of quantum reference frames beginning from a gauge theory with a unimodular symmetry group $G$. We constructed the physical Hilbert space $\mathcal{H}_{\rm phys}$ and its inner product from the kinematical Hilbert space $\mathcal{H}_{\rm kin} = \mathcal{H}_{R} \otimes \mathcal{H}_{S}$ by group averaging, and showed that conditioning physical states on a quantum reference frame $R$, taken to be a system of coherent states associated with the group $G_R$ acting on a factor $\mathcal{H}_R$, gauge fixes the physical inner product in direct analogy with the Faddeev--Popov procedure used in path integrals. The resulting map $\mathcal{R}(r)$ from physical states to conditional states of the system $S$ relative to the frame $R$, defined in Eqs.~\eqref{eq:conditionalState} and~\eqref{eq:reducationMap}, is a generalization of the Page--Wootters reduction map. Its principal virtue is that it does not require $G$ to act as a tensor product representation across the factorization  $\mathcal{H}_{\rm kin} = \mathcal{H}_{R} \otimes \mathcal{H}_{S}$, nor $G_R$ to coincide with $G$, and it therefore applies to frames that interact with the systems they describe.

Within this framework, we constructed relational Dirac observables on $\mathcal{H}_{\rm phys}$. In doing so, we generalized existing relativization maps, specfically, the quantization of classical power-series expressions for relational observables~\cite{dittrich_partial_2006,dittrich_partial_2007,thiemann_modern_2008, hohntrinity2021,hohn_equivalence_2021}, the $\$$-map of Ref.~\cite{bartlett_reference_2007}, and the \textyen-map of Ref.~\cite{loveridge_symmetry_2018}, by the inclusion of the Faddeev--Popov factor $\Delta_R(r)$ that accompanies the gauge-fixed inner product in Eq.~\eqref{eq:PhysInnerProduct2}.  When the frame and system described relative to it do not interact, the previous constructions are recovered, while in general, the Faddeev--Popov factor is what ensures the resulting operators are correctly normalized and associated with a well-defined gauge fixing. We further show that the expectation values of these observables admit a conditional probability interpretation that can be understood in terms of partial and complete observables~\cite{rovelli_partial_2002,dittrich_partial_2006,dittrich_partial_2007, rovelli_why_2014,thiemann_modern_2008}. The use of Faddeev--Popov operators and a generalization of the Page--Wootters conditional probabilities was also advocated in \cite{Chataignier:2019,Chataignier:2020}, and the relational observables constructed there can be seen as a previous version of the present construction.

Our main result is the construction of the generalized change-of-frame map $U(\tilde{r},r)$ defined in Eq.~\eqref{eq:change-frame-map} and the proof of its unitarity. Since the generalized Page--Wootters reduction maps are invertible, the composition $U(\tilde{r},r) = \mathcal{R}(\tilde{r})\circ\mathcal{R}^{-1}(r)$ relates the description supplied by any two frames, and we proved that this map is unitary under the stated conditions. The physical content of the theory is thus independent of the frame used to express it.

We applied the formalism in two settings. For temporal frames, we showed that a unitary Schr\"{o}dinger equation emerges for the system relative to the clock even when the clock and system interact, so that the clock-system coupling does not obstruct a unitary relational dynamics; in an appropriate limit, the time-nonlocal Schr\"{o}dinger equation reported in Ref.~\cite{smith_quantizing_2019} is recovered. This resolves the apparent non-unitarity of the relational dynamics observed previously for interacting clocks~\cite{smith_quantizing_2019, paiva_non-inertial_2022, paiva_flow_2022, cohen_quantum_2023, mendes_non-linear_2021, rijavec_robustness_2023, gemsheim_emergence_2023,ludescher_quantum_2026,rijavec_conditions_2026,shimada_nonunitarity_2026}, which we show to be an artifact of normalizing conditional states with the naive inner product rather than the gauge-fixed physical inner product. These two inner products differ by the Faddeev--Popov factor precisely when the action of the gauge group fails to factorize across $\mathcal{H}_{R}\otimes\mathcal{H}_S$; with the correct inner product, the time evolution of the conditional state of the system is unitary, and the previously reported loss of normalization carries no invariant meaning. For spatial frames, we applied the change-of-frame map to coupled spatial reference frames and found that the correspondence between coherence in one frame and entanglement in another~\cite{giacominiQuantumMechanicsCovariance2019,ali_ahmad_quantum_2022,de_la_hamette_entanglement-asymmetry_2022} is modified by the coupling, so that the extent to which coherence and entanglement are exchanged under a change of reference frame depends on the interaction between the frames.

Several directions merit further work. The most immediate is that our results have been obtained under the assumption that it is possible to construct the Faddev--Popov operator $\Delta_R(r)$, which requires $\mathcal{G}[E_R(r)\otimes I_S]$ to be invertible. However, this may not always be possible, or only possible on a subspace of the physical Hilbert space $\mathcal{H}_{\rm phys}$ when the kernel $\mathcal{G}[\ket{r}\!\bra{r}\otimes I_S]$ is not trivial. This has a close analogy with the Gribov problem, which prevents global gauge fixings, and in this case, one would expect the change-of-frame map to be unitary only on a subspace of $\mathcal{H}_{\rm phys}$. Related is a further generalization of the reduction map, defined in Eqs.~\eqref{eq:conditionalState} and~\eqref{eq:reducationMap}, based on covariant instruments, along the lines of Ref.~\cite{ludescher_quantum_2026}, away from rank-1 covariant POVMs defined by a family of coherent states associated with different orientation of the employed reference frame.

A second direction in the relation of the generalized Page-Wootters formalism developed here to the perspective-neutral formalism~\cite{hoehn_how_2018,vanrietvelde_change_2019,hohn_how_2020,vanrietvelde_change_2020,vanrietvelde_switching_2023,delahametteperspectiveneutral2021}. The two approaches share a great deal: both are built on a physical Hilbert space obtained by constraint quantization, both can express relational physics through Dirac observables, both admit reference frames as subsystems, and both furnish change-of-frame maps with similar compositional structure. However, the reduction maps in each approach are, a priori, distinct. The perspective-neutral reduction proceeds by trivializing the constraint with respect to a chosen reference frame, whereas the generalized Page--Wootters formalism proceeds by gauge fixing the physical inner product, which induces a generalized Page--Wootters reduction map. For ideal frames and tensor-product gauge group actions, one expects these two formalisms to agree~\cite{hohntrinity2021,hohn_equivalence_2021,delahametteperspectiveneutral2021}, and establishing explicitly when that is not the case would clarify the status of both.

Another direction concerns the relation of the present construction to the algebraic approach to quantum reference frames and to the crossed product~\cite{ jensen_generalized_2023, ali_ahmad_quantum_2024,de_vuyst_crossed_2025,de_vuyst_gravitational_2025, fewster_quantum_2025, gomez_entanglement_nodate}. Group averaging over $G$ produces an algebra of gauge-invariant observables that is related to a crossed product of the system algebra by the group action, and the Faddeev--Popov factor appearing in our gauge-fixed inner product plays a role analogous to the density operator that renders the crossed product algebra semifinite. Indeed, it would be of interest to determine whether our definition of relational observables, associated with the reduction map $\mathcal{R}(r)$, implements the passage from a Type III to a Type II algebra effected by adjoining a frame, and whether the unitarity of the change-of-frame map has an algebraic counterpart as an isomorphism of crossed product algebras built from distinct frames (see Ref.~\cite{ali_ahmad_quantum_2024} for a related discussion when $G=G_R$). A related question is what becomes of the construction when $G$ is not unimodular~\cite{ali_ahmad_semifinite_2025}. In this case, the modular function then obstructs the equivalence of the left and right Haar measures used in group averaging, and the interplay between this obstruction and the type classification of the resulting algebra of relational observables deserves attention.

Finally, given that the frame-dependence of coherence and entanglement is sensitive to interactions, it would be worth investigating how other frame-dependent notions, such as entanglement entropy, subsystem structure, and causal structure, behave under the generalized change-of-frame map. We leave these questions for future work.

\begin{acknowledgments}
We thank Nadja Magalh\~{a}es and the organizers of the 2025 Time in Physics International Workshop in S\~{a}o Paulo, Brazil, where the seeds that led to this paper were planted. We also thank Shadi Ali Ahmad for useful feedback on an early draft of this article.
\end{acknowledgments}

\appendix

\onecolumngrid

\section{The conditional system Hilbert space is a subspace of the kinematical system Hilbert space}
\label{SubspaceAppendix}

Consider a conditional state $\ket{\psi_{S}(r)} \in \mathcal{H}_{S|R}$, defined as the image of $\mathcal{R}$ acting on a physical state $\ket{\Psi} \in \mathcal{H}_{\rm phys}$,
\eq{
\ket{\psi_{S}(r)} &\ce \mathcal{R}(r)\ket{\Psi} \\
&=\bra{r} \otimes I_S \Delta_R^{\frac{1}{2}}(r) \ket{\Psi} \\
& = \sum_k \left(\bra{r} \bra{k} \Delta_R^{\frac{1}{2}}(r) \ket{\Psi}\right) \ket{k},
\label{eq:firstStep}
}
where, without loss of generality, in passing from the second to third equality we have resolved the identity $I_S = \sum_{k} \ket{k}\!\bra{k}$ in a countable, orthonormal basis $\{\ket{k}\}$ for $\mathcal{H}_S$, which is always possible because it is assumed that $\mathcal{H}_S$ is separable; this allows $\ket{\psi_S(r)}$ to be expanded in the basis $\{\ket{k}\}$ as in the last line of Eq.~\eqref{eq:firstStep}. Similarly, we can consider another conditional state defined similarly,
\eq{
\ket{\phi_{S}(r)} = \sum_k \left( \bra{r} \bra{k} \Delta_R^{\frac{1}{2}}(r) \ket{\Phi}\right) \ket{k}.
}
Then, the kinematical inner product between these states is
\eq{
\braket{\psi_{S}(r) | \phi_{S}(r)}_{\rm kin} &= \sum_{k,k'}\left [ \left(\bra{\Psi}  \Delta_R^{\frac{1}{2}}(r) \ket{r} \ket{k}\right)   \bra{k} \right] \left[ \left(\bra{r} \bra{k'} \Delta_R^{\frac{1}{2}}(r) \ket{\Phi}\right) \ket{k'}\right] \\
&= \sum_k \left [\bra{\Psi}  \Delta_R^{\frac{1}{2}}(r) \ket{r} \ket{k} \right] \left[ \bra{r} \bra{k} \Delta_R^{\frac{1}{2}}(r) \ket{\Phi}\right]\\
&=\bra{\Psi}  \left[ \Delta_R^{\frac{1}{2}}(r) \ket{r} \!   \bra{r} \otimes \left(\sum_k \ket{k}\!\bra{k}\right) \Delta_R^{\frac{1}{2}}(r) \right]\ket{\Phi}\\
&=\bra{\Psi}  \left[ \Delta_R^{\frac{1}{2}}(r) \ket{r} \!   \bra{r} \otimes I_S\Delta_R^{\frac{1}{2}}(r) \right]\ket{\Phi} \\
&= \braket{\Psi |\Phi}_{\rm phys} \\
&=: \braket{\psi_{S}(r) | \phi_{S}(r)}_{\mathcal{H}_{S|R}},
}
where the fifth equality follows from Eq.~\eqref{eq:PhysInnerProduct2}. Thus, it is seen that the kinematical inner product between two conditional states is equal to the physical inner product between the two associated physical states, which, by definition of the conditional states, is equal to the inner product on $\mathcal{H}_{S|R}$.  Note also that the image of $\mathcal{R}(r)$ contains the zero vector, and is closed under vector addition and scalar multiplication. Thus, it follows that $\mathcal{H}_{S|R} \subseteq \mathcal{H}_S$.

\section{Derivation of the time evolution operator}
\label{SchrodingerDerivation}
The time evolution operator is defined as
\eq{
U_{S|C}(t_2,t_1) &\ce \mathcal{R}(t_2) \circ \mathcal{R}^{-1}(t_1) \\
&= \left( \bra{t_2} \otimes I_S \right) \Delta^{\frac{1}{2}}(t_2) \int_G \D f \, U(g) \Delta^{\frac{1}{2}}(t_1) \left( \ket{t_1} \otimes I_S \right)
}
Then,
\eq{
\frac{\partial}{\partial t_2} U_{S|C}(t_2,t_1) &= \left( \bra{t_2} \left(-H_C\right) \otimes I_S \right) \Delta^{\frac{1}{2}}(t_2) \int_G \D g \, U(g) \Delta^{\frac{1}{2}}(t_1) \left( \ket{t_1} \otimes I_S \right)
\\
&\quad +
\left( \bra{t_2} \otimes I_S \right) \frac{d\Delta^{\frac{1}{2}}(t_2)}{dt_2} \int_G \D g \, U(g) \Delta^{\frac{1}{2}}(t_1) \left( \ket{t_1} \otimes I_S \right)\\
&= \left( \bra{t_2}  \otimes I_S \right) \left( I_C\otimes H_S + H_{\rm int}\right) \Delta^{\frac{1}{2}}(t_2) \int_G \D g \, U(g) \Delta^{\frac{1}{2}}(t_1) \left( \ket{t_1} \otimes I_S \right)
\\
&\quad +
\left( \bra{t_2} \otimes I_S \right) \frac{d\Delta^{\frac{1}{2}}(t_2)}{dt_2} \int_G \D g \, U(g) \Delta^{\frac{1}{2}}(t_1) \left( \ket{t_1} \otimes I_S \right)\\
&= \left( H_S + \int_G \D g \, \bra{t_2} \otimes I_S \left[ H_{\rm int} \Delta(t_2) + i \frac{d \Delta^{\frac{1}{2}}(t_2)}{dt_2} \Delta^{\frac{1}{2}}(t_2) \right] U(g) \ket{t_2}\otimes I_S \right) U_{S|C}(t_2,t_1) \\
&=:H_{S|C} (t_2) U_{S|C}(t_2,t_1),
\label{eq:derivativeOfTEO}
}
where $H_{S|C} (t_2)$ is defined in Eq.~\eqref{eq:HSC}. In moving from the first to second equality, we have used the fact that the $\Delta^{\frac{1}{2}}(t_2)$ is an invariant operator, so it commutes with the constraint operator $H = H_C \otimes I_S + I_C\otimes H_S + H_{\rm int}$. This means that $-H_C\otimes I_S \Delta^{\frac{1}{2}}(t_2) = \left( I_C\otimes H_S + H_{\rm int}\right) \Delta^{\frac{1}{2}}(t_2)$. In moving from the second to third equality,  we have inserted the physical identity operator $I_{\rm phys} = \int \D \, \Delta^{\frac{1}{2}}(t_2)  U(g) \left( \ket{t_2}\!\bra{t_2} \otimes I_S \right) \Delta^{\frac{1}{2}}(t_2)$. Integrating Eq.~\eqref{eq:derivativeOfTEO} yields  the formal solution to this equation
\eq{
U(t_2,t_1) = I_S - i \int_{t_1}^{t_2} \D t \, H_{S|C}(t) U(t,t_1).
}
Solving this equation iteratively via the Picard method yields the expression for the time evolution operator expressed as a time-ordered exponential,
\eq{
U_{S|C}(t_2,t_1) = \mathcal{T} \exp \left[-i\int_{t_1}^{t_2}dt \, H_{S|C}(t) \right],
}
as stated in Eq.~\eqref{eq:timeorderedExponential}.

\section{Changing spatial quantum reference frames}
\label{Changing spatial quantum reference frames}
As considered in Sec.~\ref{Sec:ChangeOfSpatialFrame}, suppose the gauge group is generated by the following constraint operator,
\eq{
    C &= P_A + P_B + P_C  + \lambda \left(  P_A \otimes P_B + 
    P_B \otimes P_C +  P_A \otimes P_C\right),
}
where $P_i$ for $i \in\{A,B,C\}$ are momentum operators on $\mathcal{H}_A \simeq L^2(\mathbb{R})$. 

The change-of-frame map in Eq.~\eqref{eq:change-frame-map} reads
\eq{
U_{A\to B}(y,x) = \int_G \D g \, \bra{y} \otimes I_{AC}\Delta_{B}^{\frac{1}{2}}(y)U(g) \Delta_{A}^{\frac{1}{2}}(x) \ket{x} \otimes I_{BC},
}
where $\ket{x}$ and $\ket{y}$ are generalized eigenstates of the position operators on $\mathcal{H}_A$ and $\mathcal{H}_B$, respectively. To evaluate the change-of-frame map, we first compute $\Delta_A$ by calculating
\eq{
&\mathcal{G}[\ket{x}\!\bra{x} \otimes I_{BC}] \\
&= \int dg \, U(g) \ket{x}\!\bra{x} \otimes I_{BC}  U(g)^\dagger\\
&= \int \D g \, e^{-igP_A} e^{-ig \lambda P_A\otimes P_B } e^{-ig \lambda P_A\otimes P_C } \ket{x}\!\bra{x} \otimes I_{BC} e^{igP_A} e^{ig \lambda P_A\otimes P_B } e^{ig \lambda P_A\otimes P_C }\\
&= \int \D g \,  e^{-ig \lambda P_A\otimes P_B }  e^{-ig \lambda P_A\otimes P_C } \ket{x+g}\!\bra{x+g} \otimes  \int dp_B \, \ket{p_B}\!\bra{p_B} \otimes  \int dp_C \, \ket{p_C}\!\bra{p_C}  e^{ig \lambda P_A\otimes P_B } e^{ig \lambda P_A\otimes P_C }\\
&=\int \D p_B \int dp_C\int dg \,  e^{-ig \lambda  \left(  p_B +  p_C \right)P_A  }  \ket{x+g}\!\bra{x+g} \otimes  \ket{p_B}\!\bra{p_B} \otimes  \ket{p_C}\!\bra{p_C} e^{ig\lambda  \left( p_B +   p_C \right)P_A } \\
&= \int \D p_B \int \D p_C \int \D g \,    \ket{x+g\left[1  + \lambda  \left(p_B +    p_C\right)\right]}\!\bra{x+g\left[1  + \lambda  \left(p_B +    p_C\right)\right]} \otimes  \ket{p_B}\!\bra{p_B}  \otimes \ket{p_C}\!\bra{p_C}.
}
To proceed, let us change integration variables to $y = x+g\left[1  + \lambda  \left(p_B +    p_C\right)\right]$, so that $dg = dy/\left[1  + \lambda  \left(p_B +    p_C\right)\right]$ and
\eq{
\mathcal{G}[\ket{x}\!\bra{x} \otimes I_{BC}]
&=  \int dp_B   \int dp_C \,  \frac{1}{1+ \lambda \left( p_B +  p_C\right)} \int dy \,    \ket{y}\!\bra{y} \otimes  \ket{p_B}\!\bra{p_B}  \otimes \ket{p_C}\!\bra{p_C}\\
&=  I_A \otimes \int dp_B \int dp_C \, \frac{1}{1+ \lambda \left(  p_B +  p_C\right)}   \ket{p_B}\!\bra{p_B}  \otimes \ket{p_C}\!\bra{p_C}\\
&=   I_A \otimes  \frac{1}{1+ \lambda \left(  P_B +   P_C\right)}.
}
Given that $\Delta_A \ce\mathcal{G}[\ket{x}\!\bra{x} \otimes I_{BC}]^{-1}$, it then follows that
\eq{
\Delta_A^{\frac{1}{2}} &= \sqrt{1+ \lambda  \left(P_B +  P_C \right)},
}
and a similar calculation yields
\eq{
\Delta_B^{\frac{1}{2}} &=  \sqrt{1+ \lambda \left(  P_A +  P_C\right) }.
}
Note that $\Delta_A$ and $\Delta_B$ are independent of the orientations $x$ and $y$ of reference frames $A$ and $B$. The change-of-frame map is then
\eq{
U_{A\to B}(y,x) &=  \int_G \D g \, \bra{y} \otimes I_{AC}\Delta_{B}^{\frac{1}{2}}(y)U(g) \Delta_{A}^{\frac{1}{2}}(x) \ket{x} \otimes I_{BC}\\
&= \frac{1}{2\pi} \int_G \D g \, \int \D p_B \,  e^{iyp_B} \bra{p_B} \otimes I_{AC}\Delta_{B}^{\frac{1}{2}}(y)U(g) \Delta_{A}^{\frac{1}{2}}(x) \int dp_A e^{-ixp_A} \ket{p_A} \otimes I_{BC}\\
&= \frac{1}{2\pi}  \int_G \D g  \int \D p_A \D p_B \, e^{i(yp_B-xp_A)}\Delta_{B}^{\frac{1}{2}} U(g) \Delta_{A}^{\frac{1}{2}}  \ket{p_A} \otimes \bra{p_B} \otimes I_C\\
&= \frac{1}{2\pi}  \int_G \D g  \int \D p_A \D p_B \, e^{i(yp_B-xp_A)}  \ket{p_A} \otimes \bra{p_B} \otimes \eta_C(p_A,p_B),
}
where
\eq{
\eta_C(p_A,p_B) &\ce\sqrt{1+ \lambda\left(  p_A +  P_C\right) } \sqrt{1+ \lambda \left(  p_B + P_C\right) } e^{-ig\left(p_A + p_B + P_C  + \lambda \left[ p_Ap_B + 
      p_B P_C + p_A P_C\right] \right)}.
}

Now suppose that from the perspective of $A$, the state of $B$ and $C$ is in a product state given by
\begin{equation}
    \ket{\psi_{BC|A}(x)} = \ket{\psi_B} \ket{\psi_C} = N\left(\ket{z_1}+ \ket{z_2} \right) \int \D p_C \, \psi_C(p_C)\ket{p_C},
\end{equation}
where $N \ce 1/\sqrt{2 +2e^{-(z_1-z_2)^2 /4\sigma^2}}$ is a normalization constant and 
\eq{
\ket{z_i} = \int \D p_B \,  e^{-ip_Bz_i} \phi(p_B) \ket{p_B}, \quad \phi(p_B) \ce \frac{\sqrt{\sigma}}{\pi^{\frac{1}{4}} } e^{-\frac{ \sigma^2 p^2}{2 }}, 
}
is a Gaussian wave packet centered at $z_i\in\mathbb{R}$ with width in position space $\sigma>0$.
Then
\eq{
\ket{\psi_{AC|B}(y)} &= U_{A\to B}(y,x) \ket{\psi_{BC|A}(x)}\\
&= \frac{N}{2\pi}\int_G \D g  \int \D p_A \D p_B \, e^{i(yp_B-xp_A)} \left(e^{-ip_Bz_1} + e^{-ip_Bz_2} \right)\phi(p_B) \ket{p_A} \otimes  \eta_C(p_A,p_B) \ket{\psi_C}\\
     &= \int \D p_A \D p_C \, \psi_{AC|B}(y,p_A,p_C)\ket{p_A}  \ket{p_C},
}
where
\eq{
\psi_{AC|B}(y,p_A,p_C) &\ce \frac{N}{2\pi}\int \D g \D p_B \, e^{-ixp_A}\left(e^{-ip_B(z_1-y)} + e^{-ip_B(z_2-y)}\right)e^{-ig\left(p_A + p_B  + p_C + \lambda \left[  p_Ap_B + 
      p_B p_C +  p_A p_C \right]\right)} \\
     &\quad\times \sqrt{1+ \lambda \left(  p_A + p_C \right) } \sqrt{1+ \lambda \left(  p_B +  p_C \right) }  \phi(p_B) \psi_C(p_C)\\
     &=  N \int \D p_B \, \left(e^{-ip_B(z_1-y)} + e^{-ip_B(z_2-y)}\right)  \delta \left(p_A + p_B  + p_C + \lambda \left[ p_Ap_B  +
      p_Bp_C + p_Ap_C\right]  \right) \\
     &\quad\times \sqrt{1+ \lambda \left(  p_A +   p_C \right)} \sqrt{1+ \lambda \left(  p_B +   p_C\right)}\phi(p_B)  \psi_C(p_C)\\
     &= N \int \D p_B \, e^{-ixp_A}\left(e^{-ip_B(z_1-y)} + e^{-ip_B(z_2-y)}\right) 
     \frac{ \delta \left(p_B-\tilde{p} \right)}{\abs{1+\lambda \left( p_A + p_C\right) }} \\
     &\quad \times \sqrt{1+ \lambda \left(p_A +   p_C\right)} \sqrt{1+ \lambda \left(  p_B +   p_C\right)}  \phi(p_B)\psi_C(p_C)\\
     &= N \frac{ \sqrt{1+ \left( p_A +   p_C\right)} }{\abs{1+\lambda \left( p_A +  p_C \right)}}   \sqrt{1+ \lambda \left(  \tilde{p} +  p_C \right)}  e^{-ixp_A} \left( e^{-i \tilde{p}(z_1-y)} + e^{-i\tilde{p}(z_2-y)}\right)  \phi(\tilde{p}) \psi_C(p_C) \\
     &= N \frac{ \sqrt{1+ \lambda p_C + \lambda^2 p_C ^2} }{\abs{1+\lambda \left( p_A +  p_C \right)}}     e^{-ixp_A} \left( e^{-i \tilde{p}(z_1-y)} + e^{-i\tilde{p}(z_2-y)}\right)  \phi(\tilde{p}) \psi_C(p_C). 
}
and
\eq{
\tilde{p} \ce - \frac{p_A+p_C + \lambda p_A p_C}{1+\lambda\left( p_A +  p_C \right)}.
}
Letting $\psi_C(p_C) = \phi(p_C)$ yields the result stated in Eq.~\eqref{eq:TransformedState} of Sec.~\ref{Sec:ChangeOfSpatialFrame}.

\bibliographystyle{apsrev4-2}        
\twocolumngrid

\bibliography{references}

\end{document}